\documentclass[11pt,a4paper]{article}
\usepackage{fullpage}
\usepackage[english]{babel}
\usepackage[T1]{fontenc}
\usepackage{mathptmx} 
\usepackage[hidelinks]{hyperref}
\usepackage{ifthen}
\usepackage{makecell}
\usepackage{multirow}
\usepackage{hhline}
\usepackage{bm}
\usepackage{threeparttable}
\usepackage{tikz}

\usepackage{pifont}
\usepackage{relsize}

\newcommand{\mparagraph}[1]{\noindent \textbf{#1} \;}
\usepackage{hyphenat}
\usepackage{enumitem}
\usepackage{graphicx}
\usepackage{amsmath,amssymb,amsthm}
\usepackage{thmtools}
\usepackage{refcount}
\usepackage{hyperref}
\usepackage[vlined,boxed,linesnumbered]{algorithm2e}

\usepackage{cleveref} 
\usepackage{url}
\usepackage{afterpage}
\usepackage[backend=biber,style=trad-plain,giveninits=true,maxbibnames=99,citetracker=true,url=false,doi=false,giveninits=true,sortcites]{biblatex}

\usepackage{etoolbox}

\newtheorem{theorem}{Theorem}[section]

\newtheorem{lemma}[theorem]{Lemma}
\newtheorem{corollary}[theorem]{Corollary}

\newtheorem{definition}[theorem]{Definition}

\DeclareMathOperator{\Forall}{\forall\,}

\newtheorem{subremark}{Remark}[theorem]
\crefname{subremark}{Remark}{Remarks}

\crefname{subobservation}{Observation}{Observations}

\newcommand{\Primes}{\ensuremath{\mathsf{Primes}}\xspace}

\newcommand{\proofATinappendix}{\em (Proof in \Cref{sec:proof-at}.)}
\newcommand{\proofQAATinappendix}{\em (Proof in \Cref{sec:proof-qaat}.)}

\newenvironment{proofsketch}{
\renewcommand{\proofname}{Proof sketch}%
\begin{proof}
}{%
\end{proof}
}%

\SetKwHangingKw{Init}{init:}
\SetKwProg{Operation}{operation}{ is}{}
\SetKwProg{When}{when}{ do}{}
\SetKwProg{Routine}{infinitely often do}{}{}
\SetKwProg{UntilDo}{until}{ do}{}
\SetKwProg{try}{try}{:}{}
\SetKwProg{catch}{catch}{:}{end}

\SetKwProg{Function}{function}{ is}{}
\SetKwProg{Predicate}{predicate}{ is}{}

\SetKwProg{InternalOperation}{internal operation}{ is}{}
\SetKwFor{ForEach}{for each}{do}{}

\SetKw{Assert}{assert}

\SetArgSty{upshape}
\SetKwComment{Comment}{$\rhd$}{}
\SetCommentSty{textit}

\newcommand{\AlgoSkip}{\vspace{0.5em}}
\SetAlCapSkip{.5em}
\SetAlgoVlined
\DontPrintSemicolon
\SetInd{0.5em}{0.6em}

\let\oldnl\nl
\newcommand{\nonl}{\renewcommand{\nl}{\let\nl\oldnl}}

\newcommand{\defeq}{\stackrel{\mathsmaller{\mathsf{def}}}{=}}

\newcommand{\mock}[1]{\widehat{#1}}

\newcommand{\Adv}{\ensuremath{\mathit{Adv}}\xspace}
\newcommand{\aangle}[1]{\langle #1 \rangle}

\newcommand{\tx}{\ensuremath{\tau}\xspace}

\newcommand{\supportSet}{\ensuremath{\mathsf{set}}\xspace}
\newcommand{\transferSet}{\ensuremath{\mathsf{transferSet}}\xspace}

\newcommand{\transferFunc}{\ensuremath{\mathit{transfer}}\xspace}

\newcommand{\sendingUA}{\ensuremath{\mathit{sendingUA}}\xspace}
\newcommand{\receivingUA}{\ensuremath{\mathit{receivingUA}}\xspace}
\newcommand{\precedingUAs}{\ensuremath{\mathit{precedingUAs}}\xspace}

\newcommand{\precedingTrans}{\ensuremath{\mathit{precedingTrans}}\xspace}

\newcommand{\bal}{\ensuremath{\mathit{bal}}\xspace}

\newcommand{\init}{\ensuremath{\mathtt{init}}\xspace}

\newcommand{\transferm}{\textsc{transfer}\xspace}

\newcommand{\prover}{\ensuremath{\mathit{pvr}}\xspace}

\newcommand{\balcom}{\ensuremath{\mathit{bal\_c}}\xspace}
\newcommand{\balopen}{\ensuremath{\mathit{bal\_o}}\xspace}

\newcommand{\balancee}{\ensuremath{\mathsf{balance}}\xspace}
\newcommand{\transfer}{\ensuremath{\mathsf{transfer}}\xspace}

\newcommand{\total}{\ensuremath{\mathit{total}}\xspace}

\newcommand{\prove}{\textup{\ensuremath{\mathsf{process\_transfer}}}\xspace}

\newcommand{\snd}{\ensuremath{\mathit{snd}}\xspace}
\newcommand{\rcv}{\ensuremath{\mathit{rcv}}\xspace}

\newcommand{\TransferValidity}{\ensuremath{\mathit{TransferValidity}}\xspace}
\newcommand{\AccountUpdate}{\ensuremath{\mathit{AccountUpdate}}\xspace}

\newcommand{\ap}{\ensuremath{\sigma}\xspace}
\newcommand{\Pa}{\ensuremath{\mathit{P_A}}\xspace}

\newcommand{\approve}{\textup{\ensuremath{\mathsf{ap\_prove}}}\xspace}
\newcommand{\apverify}{\textup{\ensuremath{\mathsf{ap\_verify}}}\xspace}

\newcommand{\data}{\ensuremath{\mathit{data}}\xspace}
\newcommand{\sn}{\ensuremath{\mathit{sn}}\xspace}

\newcommand{\sig}{\ensuremath{\mathit{sig}}\xspace}
\newcommand{\sigs}{\ensuremath{\mathit{sigs}}\xspace}

\newcommand{\seqnums}{\ensuremath{\mathit{seq\_nums}}\xspace}

\newcommand{\quoruminitm}{\textsc{quorumInit}\xspace}
\newcommand{\quorumsigm}{\textsc{quorumSig}\xspace}

\newcommand{\comCommit}{\ensuremath{\mathsf{c\_commit}}\xspace}
\newcommand{\comVerify}{\ensuremath{\mathsf{c\_verify}}\xspace}

\renewcommand{\mp}{\ensuremath{w}\xspace}
\newcommand{\nmp}{\ensuremath{u}\xspace}

\newcommand{\accIsEmpty}{\textup{\ensuremath{\mathsf{ua\_is\_empty}}}\xspace}
\newcommand{\accAdd}{\textup{\ensuremath{\mathsf{ua\_add}}}\xspace}

\newcommand{\accProveMem}{\textup{\ensuremath{\mathsf{ua\_prove\_mem}}}\xspace}
\newcommand{\accProveNonMem}{\textup{\ensuremath{\mathsf{ua\_prove\_non\_mem}}}\xspace}
\newcommand{\accVerifyNonMem}{\textup{\ensuremath{\mathsf{ua\_verify\_non\_mem}}}\xspace}
\newcommand{\accVerifyMem}{\textup{\ensuremath{\mathsf{ua\_verify\_mem}}}\xspace}

\newcommand{\zkp}{\ensuremath{\pi}\xspace}
\newcommand{\Pzk}{\ensuremath{\mathit{P_{ZK}}}\xspace}

\newcommand{\secdata}{\ensuremath{\mathit{secret\_data}}\xspace}
\newcommand{\pubdata}{\ensuremath{\mathit{public\_data}}\xspace}
\newcommand{\olddata}{\ensuremath{\mathit{old\_state\_data}}\xspace}

\newcommand{\zkpprove}{\textup{\ensuremath{\mathsf{zkp\_prove}}}\xspace}
\newcommand{\zkpverify}{\textup{\ensuremath{\mathsf{zkp\_verify}}}\xspace}

\newcommand{\send}{\textup{\ensuremath{\mathsf{send}}}\xspace}
\newcommand{\receive}{\textup{\ensuremath{\mathsf{receive}}}\xspace}
\newcommand{\receivedd}{\textup{\ensuremath{\mathsf{received}}}\xspace}

\newcommand{\rasend}{\textup{\ensuremath{\mathsf{ra\_send}}}\xspace}
\newcommand{\rareceive}{\textup{\ensuremath{\mathsf{ra\_receive}}}\xspace}
\newcommand{\rareceived}{\textup{\ensuremath{\mathsf{ra\_received}}}\xspace}

\newcommand{\broadcast}{\textup{\ensuremath{\mathsf{broadcast}}}\xspace}

\newcommand{\wait}{\textup{\ensuremath{\mathsf{wait}}}\xspace}
\newcommand{\return}{\textup{\ensuremath{\mathsf{return}}}\xspace}

\newcommand{\msgm}{\textsc{msg}\xspace}

\newcommand{\sequence}{\ensuremath{S}\xspace}

\newcommand{\op}{\ensuremath{\mathsf{op}}\xspace}

\newcommand{\commit}{\ensuremath{\mathtt{commit}}\xspace}
\newcommand{\abort}{\ensuremath{\mathtt{abort}}\xspace}

\newcommand{\ffalse}{\ensuremath{\mathtt{false}}\xspace}
\newcommand{\ttrue}{\ensuremath{\mathtt{true}}\xspace}

\newcommand{\SpacedCommentLine}[1]{\vspace{0.1em}\par\Comment{\nonl#1}}

\newcounter{mtassumctr}

\newcommand{\ie}{\textit{i.e.},\xspace}
\newcommand{\eg}{\textit{e.g.},\xspace}

\newcommand{\POATSeq}{\rightsquigarrow_i}%
\newcommand{\POATSeqOne}{\overset{\textbf{1}}{\rightsquigarrow}_i}%
\newcommand{\POATSeqTwo}{\overset{\textbf{2}}{\rightsquigarrow}_i}%

\newcommand{\transTo}[1]{t_{\vartriangleright #1}}%
\newcommand{\transFrom}[1]{t_{#1 \vartriangleright}}%

\newif\ifannote

\ifannote
    \newcommand{\anninsert}[2]{{\color{#1}#2}}
    \newcommand{\anncomment}[3]{
        {\color{#1}\colorbox{#1}{\bfseries\sffamily\tiny\textcolor{white}{#2}}
        $\blacktriangleright$ \em #3 $\blacktriangleleft$}
    }
\else
    \newcommand{\anninsert}[2]{#2}
    \newcommand{\anncomment}[3]{}
\fi

\newcommand{\TA}[1]{\anncomment{blue}{TA}{#1}}

\newcommand{\ft}[1]{\anninsert{magenta}{#1}}
\newcommand{\ta}[1]{\anninsert{blue}{#1}}
\newcommand{\ar}[1]{\anninsert{teal}{#1}}

\newcommand{\todo}[1]{\anncomment{red}{TODO}{#1}}

\bibliography{bibliography}

 \title{\textbf{Asynchronous BFT Asset Transfer: \texorpdfstring{\\}{} Quasi-Anonymous, Light, and Consensus-Free}}

 \author{Timoth\'e Albouy,
    ~Emmanuelle Anceaume,
    ~Davide Frey, \\
    ~Mathieu Gestin,
    ~Arthur Rauch,
    ~Michel Raynal,
    ~Fran\c{c}ois Ta\"{i}ani
    ~\\~\\
 Univ Rennes, Inria, CNRS, IRISA, 
 35042 Rennes-cedex, France
}

\begin{document}

\maketitle

\begin{abstract}
This paper introduces a new asynchronous Byzantine-tolerant asset transfer system (cryptocurrency) with three noteworthy properties: quasi-anonymity, lightness, and consensus-freedom. Quasi-anonymity means no information is leaked regarding the receivers and amounts of the asset transfers. Lightness means that the underlying cryptographic schemes are \textit{succinct} (\textit{i.e.}, they produce short-sized and quickly verifiable proofs) and each process only stores its own transfers while keeping communication cost as low as possible. Consensus-freedom means the system does not rely on a total order of asset transfers. The proposed algorithm is the first asset transfer system that simultaneously fulfills all these properties in the presence of asynchrony and Byzantine processes. To obtain them, the paper adopts a modular approach combining a new distributed object called ``agreement proof'' and well-known techniques such as commitments, universal accumulators, and zero-knowledge proofs. 

\end{abstract}

\mparagraph{Acknowledgments.}
The authors would like to thank Maxence Brugères, Petr Kuznetzov, and Victor Languille for their valuable comments which helped improve this paper significantly.

This work is partially supported by the French ANR projects ByBloS (ANR-20-CE25-0002-01) and PriCLeSS (ANR-10-LABX-07-81), and the SOTERIA project.
SOTERIA has received funding from the European Union's Horizon 2020 research and innovation programme under grant agreement No 101018342.
This content reflects only the author's view.
The European Agency is not responsible for any use that may be made of the information it contains.

\todo{Changer les refs pour les thresholds signatures pour mettre celles basées sur SSS (DKG) -> Arthur}


\todo{RELIRE INTRO: Potentiellement ajouter une footnote que notre façon de noter et de raisonner autour des ZKP est nouvelle et qu'elle est utilisée spécifiquement pour le distribué. + On est les premiers à étuddié Byzantins + asynchrone + ZKP etc... -> Francois + Davide }


\todo{To check: Make sure that we distinguish invocations (of operations) and operations themselves.}
\todo{Final version: thanks Maxence, but not in submission for anonymity reason.}

\thispagestyle{empty}
\setcounter{page}{0}
\clearpage



\section{Introduction} \label{sec:intro}

Imagine a world without the ability to send money instantly across the globe, a world where financial borders seemed impassable.
This was our world just a few decades ago.
The advent of online money transfers, a real silent revolution, has turned our lives upside down and is reshaping the global economic landscape.
Its potential for a positive impact on the world (\eg decentralized and secured financial transactions, democratized access to financial services, micro-payments) is far from exhausted. 
Existing decentralized money transfer, also called \textit{decentralized  or  distributed asset transfer},
faces a range of challenges pertaining to  privacy (\ie transfer confidentiality), performance (\ie throughput and latency), and communication/computational/storage costs.
Addressing these challenges amounts to designing a system that enjoys:
\begin{itemize}[topsep=0pt,itemsep=-1ex,partopsep=1ex,parsep=1ex]
\item \textbf{Anonymity}  to hide both the identity of the parties involved in the transfer (\ie the sender and/or the recipient)\footnote{Note that when the anonymity of only the sender or the recipient is hidden, we talk about quasi-anonymity.} and the amount of transferred assets,
\item \textbf{Lightness} to maintain low operational costs in terms of computation (in particular, ensuring the \emph{succinctness} of security proofs), storage (each participant only needs to store her own transfer history), and communication (\ie low bit complexity),
\item \textbf{Consensus freedom} to operate deterministically in an asynchronous failure-prone setting.
\end{itemize}

\noindent
Unfortunately, existing systems invariably satisfy only a subset of these three important properties. 

First, until recently, all asset-transfer systems relied on a consensus primitive.
Consensus was thought necessary to avoid double spending, but as shown in \cite{AFRT20,GKMPS22,G16}, a reliable broadcast primitive whose deliveries respect process sending order is sufficient when each account is owned by a single process.
This has direct practical implications because consensus imposes a non-required sequential processing of all the asset transfers (or block by block in the case of blockchains).
On the other hand, not only can consensus-free asset-transfer systems process transfers concurrently, resulting in higher throughput~\cite{BDS20,CGKKMPPST20}, but they also make it possible to leverage deterministic algorithms that achieve progress in asynchronous settings.
In practical systems, this means not having to wait for synchrony assumptions to be satisfied, a condition that would otherwise lead to long delays in large-scale systems.
The interested reader may find possibility and impossibility results for payment channels in asynchronous asset transfer in \cite{NK22}.

Second, even if recent years have shown the emergence of a variety of consensus-free solutions for asset transfer, none of them, except for the recent Zef~\cite{BSKD23} feature any level of anonymity.
Anonymous asset transfer systems typically rely on consensus to enable multiple users to hide behind a set of accounts or tokens.
Indeed, recent work~\cite{FGR23} has shown that consensus is necessary to guarantee full anonymity in a system where asset transfers from correct processes never fail. Zef~\cite{BSKD23} circumvents this impossibility by weakening the anonymity requirement, similar to what we do in this paper. 

Finally, most systems (including Zef~\cite{BSKD23}) exhibit high storage and/or communication cost as summarized in Table~\ref{tab:comparison}.
With respect to storage, most systems require nodes to store the entire history of all asset transfers (\eg the entire blockchain) or at least a number  of transfers proportional to this entire history. With respect to communication, they feature a cost that is at least linear in the number of nodes, and quadratic in the case of Zef (Table~\ref{tab:comparison}).

\mparagraph{Combining quasi-anonymity, lightness, and consensus freedom.}
The challenge  addressed in this paper consists in designing the first asset transfer system that achieves anonymity, and consensus-freedom while incurring as-low-as-possible storage and communication costs. Specifically, the proposed solution  achieves the following properties.

\begin{enumerate}[topsep=0pt,itemsep=-1ex,partopsep=1ex,parsep=1ex]
    \item \textbf{Quasi-anonymity:} it hides the amount and the receiver's identity of each asset transfer.
    
    
    
    \item \textbf{Lightness:} 
    All the cryptographic schemes it uses are succinct---\ie with at most polylogarithmic proof size and verification time---the storage cost incurred by each process $p_i$ is linear in the number of transfers of $p_i$ for a fixed security parameter $\lambda$, and the communication cost of the algorithm remains as low as possible.
    
    \item \textbf{Consensus freedom:} The proposed solution consists of a deterministic algorithm that can operate in an asynchronous setting prone to failures, thereby supporting responsive implementations. 
\end{enumerate}



\mparagraph{Roadmap.}
This paper is made up of \ref*{sec:conclusion} 
main  sections.
\Cref{sec:background} positions the contribution with respect to the state of the art. 
\Cref{sec:model} defines the computing system model.
\Cref{sec:amt-spec} provides a concurrent specification of quasi-anonymous asset transfer (QAAT).
\Cref{sec:techniques} informally defines the cryptographic schemes used in our system.
\Cref{sec:amt-algo} presents a novel 
{\it Agreement Proof} abstraction, our QAAT system, and the intuition behind its correctness.
Finally, \Cref{sec:conclusion} concludes the article.

Mote technical developments, such as the proofs of 
the QAAT algorithm, and the properties and implementations of \Cref{sec:techniques}'s s
chemes appear in appendices.

\section{Background and State of the Art}
\label{sec:background}
As we briefly mentioned above, some notable systems have addressed each of the challenges we presented individually, but to the best of our knowledge, we are the first to address all of them in a single system.
\Cref{tab:comparison} compares some of the notable existing asset transfer systems evoked below.

\begin{table}[t]
    \centering
    \begin{tabular}{|c||c|c|c|c|c|}
    \hline
    {\small \textbf{System}} & {\small \textbf{Anonymity}} & {\small \textbf{Comm.}} & {\small \textbf{Storage}} & {\small \textbf{Cons.-Free}} & {\small \textbf{Model}} \\
    \hhline{|=#=|=|=|=|=|}
    {\small Zcash~\cite{HBHW16}} & Full & $\Omega(\lambda n)$ & $\Omega(\lambda|T|)$ & No & Sync. \\
    \hline
    {\small Mina~\cite{BMVS20,SJSW20}} & None & $\Omega(\lambda n)$ & $\Omega(\lambda|T|/n + n \log n + n)$ & No & Sync. \\
    \hline
    \makecell{{\small Pastro~\cite{KPPT23}}} & None & ? & $\Omega(\lambda|T| + n)$ & Yes & Async. \\
    \hline
    \makecell{{\small Zef~\cite{BSKD23}}} & Quasi & $\Omega(\lambda n^2)$ & $\Omega(\lambda|T| + n)$ \footnotemark & Yes & Async. \\
    \hline
    {\small \textbf{This paper}} & Quasi & $O(\lambda n)$ & $O(\lambda+(|T|/n)\log n + n)$ & Yes & Async. \\
    \hline
    \end{tabular}
    \vspace{.5em}
    \caption{Comparison  with notable existing systems.
    {$|T|$ denotes the total number of transfers  in an execution.
The storage cost is given per process, and the communication cost is given for the whole system.}
    }
    \label{tab:comparison}
\end{table}

    
    



\mparagraph{Anonymous and quasi-anonymous asset transfer systems.}
Anonymous Asset Transfer (AAT) systems have been studied since the introduction of the online payment system Bitcoin in $2008$~\cite{N08} and can be categorized as either \emph{token-based} or \emph{account-based}.
The two types of systems differ by the persistency of the financial medium. A token can only be used once while an account has a recorded balance. 
To transfer a token, the sender must prove that it possesses it and that this token has never been used before.
This is generally achieved by maintaining a list of unspent tokens and a list of spent tokens.
Anonymity is typically obtained by relying on cryptographic primitives (Zero-knowledge proofs---ZKPs) that allow one party to convince another party that a given statement is true without divulging any information beyond the fact that the statement is true.
An example of token-based systems is Zcash~\cite{HBHW16}.
\footnotetext{
    The authors of~\cite{BSKD23} describe a public-key rotation method to further reduce Zef's storage cost, however this mechanism requires additional synchrony assumptions, as discussed in \Cref{sec:enhancements}. 
}

Asset transfer via accounts requires the sender to prove that its account has sufficient funds.
Anonymity is guaranteed by having the sender select at random a subset of its accounts.
The Quisquis~\cite{FMMO19} and the Zether~\cite{BAZB20} asset-transfer systems represent account-based AATs.
%

All the systems we mentioned~\cite{HBHW16,FMMO19,BAZB20} provide full anonymity (\ie they ensure both  sender and receiver anonymity), but they also rely on consensus and therefore cannot be implemented deterministically in an asynchronous failure-prone system.
Indeed, it was recently shown~\cite{FGR23} that fully anonymous asset transfer requires consensus or equivalently total order to be implemented, if asset transfers from a correct process must never fail.

We show in the this paper that guaranteeing quasi-anonymity in place of anonymity is sufficient to circumvent this impossibility result, similar to what is done by Zef~\cite{BSKD23}.

\mparagraph{Light asset transfer systems.}
The idea of lightness consists in guaranteeing that all underlying cryptographic schemes are \textit{succinct}, that processes only store transactions they are involved in, and that communication cost remains as low as possible.
A cryptographic proving scheme is succinct if and only if its proof size and verification time are at most polylogarithmic in the ``size of the problem''~\cite{BCIOP22}, where the size of the problem depends on the scheme.\footnote{Following the literature, we only consider the internal state and the cryptographic schemes used by our system when analyzing its storage cost.
In particular, the storage cost induced by the network routing protocol is ignored.
}
For instance, it may refer to the number of signatures in an aggregated signature scheme~\cite{BLS04}, or to the arithmetic circuit size used in a Succinct Non-interactive Argument scheme (SNARG).
Intuitively, succinctness captures the fact that a practicable system has to use cryptographic implementations that are themselves practicable (\ie their verification and storage costs do not become prohibitively high with time).
%

Reducing the local storage cost of decentralized asset transfer systems has been a research challenge since the advent of the Bitcoin blockchain. The notion of light clients, implemented with the simplified payment verification (SPV) protocol~\cite{N08} allows clients to only store block headers and Merkle proofs in place of full blocks.
Still, their local storage grows linearly with the size of the blockchain.
Other solutions, such as the one presented in~\cite{BKLZ20}, allow a succinct and secure construction of proofs, \ie Merkle mountain ranges, but still require the full blockchain to be maintained to update the proof with new blocks.
In contrast, non-interactive proofs of proof-of-work succeed in constructing secure and succinct token-based asset transfer systems by sampling a polylogarithmic number of blocks of the blockchain~\cite{JAG23,KLZ21} but do not target full anonymity of the transfers.
Finally, the Mina system~\cite{BMVS20} leverages SNARKs (Succinct Non-interactive Arguments of Knowledge) to obtain a succinct blockchain. 
However, their construction requires maintaining the full knowledge of all the accounts of the system to update the proof when a transfer occurs.
In addition, Mina does not provide any level of anonymity for the transfers.
Those solutions do not explicitly say that they rely on P2P for their communication. Assuming that this is the case, sending a piece of information costs $n \log(n)$ messages in average.

\mparagraph{Consensus-free asset transfer systems.} 
Fundamentally, an asset-transfer system must be double-spending-free; that is, spending a unit of value more than once must be impossible.
This means that all parties agree on the order in which transfers are issued from each individual account: if some account issues two conflicting transfers spending the same funds, the rest of the network, and especially the corresponding creditors, have to agree on which transfer (if any) is correct.
This relaxed, per-account ordering of transfers can be obtained by Byzantine-tolerant communication primitives weaker than consensus, \eg \emph{Byzantine-tolerant reliable broadcast}~\cite{B87}.%
\footnote{
    This is only the case for asset transfer systems where each account is owned by a single process.
    In systems where accounts can have multiple owners, total ordering is needed~\cite{GKMPS22,G16}, \ie consensus or atomic broadcast.
}
Following~\cite{GKMPS22}, several works have proposed payment systems based on reliable broadcast only, namely AFRT20~\cite{AFRT20}, Astro~\cite{CGKKMPPST20}, and FastPay~\cite{BDS20}.
%
Astro was later extended to permissionless environments with Pastro~\cite{KPPT23}, combining weighted quorums and proof-of-stake.
Building on FastPay, Zef~\cite{BSKD23} provides quasi-anonymous transfers, but it incurs quadratic communication cost and storage cost linear in the total number of transfers.


\section{System Model} \label{sec:model}
\mparagraph{Processes.}
\sloppy
The system comprises $n$ sequential asynchronous processes denoted by $p_1,...,p_n$.
Each process $p_i$ has a unique identity, which is known to other processes. 
To simplify the presentation without loss of generality, we assume that $i$ is the identity of $p_i$.
%
Up to $t<n/3$ processes can be Byzantine, where a Byzantine process is a process whose behavior may not follow the algorithm~\cite{LSP82,PSL80}.\footnote{The assumption $t<n/3$ is only needed for the \textit{Agreement Proof} scheme of our system (see \Cref{sec:agreement-proof}).}
Byzantine processes may collude to fool non-Byzantine (a.k.a. correct) processes.

\mparagraph{Network.}
Processes communicate by exchanging messages through a fully connected asynchronous point-to-point communication network, which is assumed to be reliable (\ie the network does not corrupt, drop, duplicate, or create messages).
Let \msgm be a message type and $v$ the value contained in the message.
The operation ``\send $\msgm(v)$ {\bf to} $p_j$'' is used for sending, and the callback ``{\bf when} $\msgm(v)$ {\bf is} \receivedd'' is used for receiving.
For syntactic sugar, processes can also communicate using the macro-operation denoted \broadcast $\msgm(v)$, that is a shorthand for ``{\bf for all} $j \in \{1, \cdots, n\}$ {\bf do} \send $\msgm(v)$ {\bf to} $p_j$''.
When processes use this macro-operation to disseminate a message, we say that this message is \textit{broadcast} and \textit{received}.
The macro-operation \broadcast $\msgm(v)$ is unreliable.
For example, if the invoking process crashes during its invocation, an arbitrary subset of processes receives the message $\msgm(v)$.
Moreover, due to its very nature, a Byzantine process can send conflicting messages without using the macro-operation \broadcast.
Finally, the processes have access to \rasend/\rareceive operations (for ``receiver-anonymous send/receive''), which function like the classical \send/\receive operations with the additional guarantees that \emph{(i)} the message cannot be read by anyone other than the receiver, and that \emph{(ii)} the receiver remains anonymous from an adversary eavesdropping the network.
For instance, these two operations could be implemented by broadcasting an encrypted message to the whole network that only the intended receiver can decrypt, or by using onion routing~\cite{RSG98}.\footnote{
    Note that both these solutions have a communication complexity of at most $O(\lambda n)$.
}

\mparagraph{Cryptographic schemes.}
In this paper, we use cryptographic schemes such as secure hash functions, asymmetric signatures, commitments, accumulators, and arguments of knowledge.
In the following, we specify these schemes regarding abstract operations and guarantees.
We provide concrete implementation examples in \Cref{sec:impl}, detailing how each choice constrains the adversary's power.

\mparagraph{Security parameter $\bm{\lambda}$.}
We denote by $\lambda$ the security parameter of the cryptographic schemes used in our algorithms.
There is a trade-off between the security level of our cryptographic schemes (represented by
$\lambda$) and their computational, communication, and storage costs.
We further assume that $\lambda$ dominates $\log n$, \ie $\lambda=\Omega(\log n)$.



\mparagraph{Different notions of adversaries.}
We consider an adversary denoted \Adv with both \textit{distributed} and     \textit{cryptographic} aspects.
The distributed aspect seeks to compromise the correctness of the system by controlling process faults (Byzantine faults in our case) and the scheduling of messages (asynchrony).
The cryptographic aspect, on the other hand, strives to compromise the anonymity properties of the system.
In this respect, we consider a \textit{probabilistic polynomial-time} (PPT) adversary, \ie one that has bounded computational power.
It has only access to the information publicly transiting on the network (\ie the messages communicated using the clear-text primitives \send, \broadcast, and \receivedd), not the messages privately exchanged between pairs of processes using the \rasend/\rareceive operations defined before.

%


\mparagraph{Notations.}
The \textit{invocation} by a process of an operation $\op$ with input parameter $\mathit{in}$ and output result $\mathit{out}$ is denoted 
$\op(\mathit{in})/\mathit{out}$.
The ``$\star$'' symbol means that the corresponding value is left unspecified, \eg $\op_i(\star)$ refers to an invocation of \op by $p_i$ with an arbitrary input.
A pair made up of $a$ and $b$ is denoted $\langle a,b \rangle$. 
\Cref{tab:notations} summarizes key notations used throughout this paper.

\begin{table*}[t]
    \centering
    \begin{tabular}{|c|c|}
    \hline
    \textbf{Notation} & \textbf{Meaning} \\
    \hhline{|=|=|}   
    $\lambda$, $\epsilon(\lambda)$, \Adv & Security parameter, Negligible real number in $[0;1]$, Adversary \\
    \hline
    AT, QAAT & Asset transfer object, Quasi-anonymous asset transfer object \\
    \hline
    AP, \ap, \Pa & Agreement proof scheme, Agreement proof, AP predicate \\
    \hline
    $c$, $o$ & Commitment digest, Opening \\
    \hline
    \makecell{UA, $A$, \\ \mp, \nmp} & \makecell{Universal accumulator scheme, Accumulator, \\Membership proof, Non-membership proof} \\
    \hline
    ZKP, \zkp, \Pzk & Zero-knowledge proof scheme, Zero-knowledge proof, ZKP predicate \\
    \hline
    $\tx{:}\aangle{\snd,v,\rcv,\sn}$ & Asset transfer details: sender, amount, receiver, sequence number \\
    \hline
    \end{tabular}
    \vspace{-.2em}
    \caption{Key notations used in the paper.} \label{tab:notations}
\end{table*}


\section{Quasi-Anonymous Asset Transfer (QAAT): Concurrent Specification} \label{sec:amt-spec}
\textit{Asset Transfer} (\textit{AT}), and its extension, \textit{Quasi-Anonymous Asset Transfer} (\textit{QAAT}) are derived from the {\it concurrent} asset transfer specification of AFRT20~\cite{AFRT20}.
We refer to this specification as concurrent because it does not consider asset transfer as a sequential object handling operation invocations in a total order.
In the following, we assume that each account is owned by a single process.

\mparagraph{AT operations.}
A distributed AT object provides processes with two operations: 
\begin{itemize}[topsep=0pt,itemsep=-1ex,partopsep=1ex,parsep=1ex]
    \item $\balancee_i()/v$ returns the account balance $v$ of the calling process $p_i$ according to its current local vision.
    
    \item $\transfer_i(v,j)/r$ transfers the amount $v$ from the account of the calling process $p_i$ to the account of $p_j$, returns $r=\commit$ if the account of $p_i$ has enough funds, $r=\abort$ otherwise.
    For simplicity, if $r=\commit$, we omit the writing of the result of the corresponding invocation: $\transfer_i(v,j)/\commit$ can be simply written $\transfer_i(v,j)$.
\end{itemize}

The above $\balancee()$ operation is weaker than that of AFRT20~\cite{AFRT20} as it can only return the balance of the local process, and not of any process of the system.
This is needed to ensure confidentiality (since a given process should not be able to read the accounts of other processes). 

It is assumed that the account of each process $p_i$ is initialized to a nonnegative value denoted $\init_i$ (in our QAAT implementation of \Crefrange{alg:aat-init-vars}{alg:aat}, $\init_i$ is known only by $p_i$).

\mparagraph{
Histories and sequences.}
\begin{itemize}[topsep=0pt,itemsep=-1ex,partopsep=1ex,parsep=1ex]
   
    \item Let   $S=(o_k)_{k\in[1..|S|]}$  be a sequence of operation invocations. $S$ might be finite (in which case $|S|\in\mathbb{N}$), or infinite (in which case $|S|=+\infty$, and $[1..|S|]=\mathbb{N}$). We note $\supportSet(S)$ the set of invocations in  $S$, $\rightarrow_S$ the total order defined by $S$, and $\transferSet(S)$ the set of \transfer invocations contained in $S$ (\ie $\transferSet(S)=\{k \in [1,|S|] \mid o_k=\transfer_\star(\star,\star)\}$).
    \item A local history of a correct process $p_i$, denoted $L_i$, is a sequence of operation invocations made by $p_i$.
    If an invocation $o$ precedes another invocation $o'$ in $L_i$, we say that ``$o$ precedes $o'$ in the process order of $p_i$'', which is written $o \rightarrow_i o'$.
    
    \item In a similar way, a local history of a Byzantine process $p_i$ is a sequence $L_i$ of operations invocations. However, because Byzantine processes might deviate arbitrarily from their prescribed behavior, these invocations do not necessarily correspond to how the behavior of Byzantine processes is \emph{perceived} by correct processes.
    %
    
    \item A global history $H$ is an array of $n$ local histories, one for each process: 
    $H = [L_1,...,L_n]$. 

    \item We use the notion of \emph{mock history} to capture the perceived behavior of Byzantine processes~\cite{AFRT20}. Intuitively, a mock history $\mock{H}$ of some global history $H$ is a history that preserves the local histories $L_i$ of correct processes $p_i$ but might change the local histories of Byzantine processes. More formally, $\mock{H} = [\mock{L}_1,...,\mock{L}_n]$ is a mock history of $H=[L_1,...,L_n]$ iff $\mock{L}_i=L_i$ for all correct processes.
    
    

  
\end{itemize}

\mparagraph{Active processes.}
Given a global history $H$, we say that a process is \emph{active} in $H$ if it is involved in a $\transfer_i(v,j)$ operation invocation appearing in a local history $L_i$ of $H$, either as the calling process $p_i$ or as the receiver $p_j$.

\mparagraph{AT-sequence.}
Given a process ID $i$ and a set of operation invocations $O$ (performed by $p_i$ and other processes), the function $\total(i,O)$ returns the balance of the account of $p_i$ resulting from the transfers it sends and receives according to $O$ (by adding the initial balance of $p_i$ to the funds received by $p_i$ in $O$, and subtracting the funds sent by $p_i$ in $O$):
\newcommand{\definitionTotal}[1]{%
\total(#1,O) = \init_{#1} + \big(\Sigma_{\transfer_\star(v,#1) \in O} \, v\big) - \big(\Sigma_{\transfer_{#1}(v,\star) \in O} \, v\big).%
}
$$
    \definitionTotal{i}
$$


\noindent
A sequence $S$ of invocations 
is an \emph{asset-transfer-sequence} (\emph{AT-sequence}) if and only if:
\begin{enumerate}[topsep=0pt,itemsep=-1ex,partopsep=1ex,parsep=1ex]
    \item $\Forall i\hspace{0.8em}\in[0..n], \Forall o = \balancee_i()/v \hspace{0.5em}\in S: v = \total\big(i,\{o' \in S \mid o' \rightarrow_S o\}\big)$\label{at-seq:balance:inv};
    
    \item $\begin{aligned}[t]
    \Forall i,j\in[0..n], \Forall o = \transfer_i(v,j) &\in S: v \leq \total\big(i,\{o' \in S \mid o' \rightarrow_S o\}\big).
    \end{aligned}$\label{at-seq:positive:account}
\end{enumerate}
Informally, these conditions mean that the \balancee operation must return the balance of the process's account when it is invoked in the sequence, and the \transfer operation must succeed (\ie return \commit) only if the balance of the debtor's account is sufficient according to the sequence $S$.

We say that a global history $H$ can be \emph{AT-sequenced} iff for every correct process $p_i$ there exists an \emph{AT-sequence} $\sequence_i$ with the following properties:
\begin{itemize}[topsep=0pt,itemsep=-1ex,partopsep=1ex,parsep=1ex]
    \item $\sequence_i$ contains exactly all of $p_i$'s invocations (both $\transfer_i$ and $\balancee_i$ invocations), and all transfers by other processes, \ie $\supportSet(\sequence_i)=\supportSet(L_i)\cup\bigcup_{j\neq i}\transferSet(L_j).$
    \item $\sequence_i$ respects $p_i$'s process order, \ie $\rightarrow_i\subseteq \rightarrow_{\sequence_i}$.
\end{itemize}
Intuitively, the above conditions ensure that each process $p_i$ can explain its local execution (Condition~\ref{at-seq:balance:inv} of the AT-sequence definition above) from the transfers contained in the system in a way that respects $p_i$'s local order and the invariant that no account should ever be negative (Condition~\ref{at-seq:positive:account} of the AT-sequence definition).


\mparagraph{AT properties.}
A distributed algorithm $\mathcal{A}$ implements the AT specification iff it provides \balancee and \transfer operations as defined earlier, such that the following properties hold.
\begin{itemize}[topsep=0pt,itemsep=-1ex,partopsep=1ex,parsep=1ex]
    \item \textbf{AT-Termination.} All operation invocations of $\mathcal{A}$ (\balancee and \transfer) by correct processes terminate.
    \item \textbf{AT-Sequentiality.} For any global history $H$ capturing an execution of $\mathcal{A}$, there exists a mock history $\mock{H}$ of $H$ that can be AT-sequenced.
\end{itemize}

\mparagraph{Quasi-Anonymous Asset Transfer.}
An algorithm $\mathcal{A}$ implements a Quasi-Anonymous AT object (QAAT for short) if it verifies the AT properties (stated above) and also meets the following privacy-preserving properties.
Recall that \Adv denotes an adversary trying to guess private information from system asset transfers.
\begin{itemize}[topsep=0pt,itemsep=-1ex,partopsep=1ex,parsep=1ex]
    
    \item \textbf{QAAT-Receiver-Anonymity.}
    For any global history $H$ capturing an execution of $\mathcal{A}$, for any invocation $\transfer_i(\star,j)$ contained in $H$ where $p_i$ and $p_j$ are correct processes, then, with high probability, $p_j$ is indistinguishable among all the correct active processes of $H$.
    

    
    \item \textbf{QAAT-Confidentiality.}
    For any global history $H$ capturing an execution of $\mathcal{A}$, for any invocation $\transfer_i(v,j)$ contained in $H$ where $p_i$ and $p_j$ are correct processes, then, with high probability, $v$ is indistinguishable among any value in $\mathbb{R}^+$.


\end{itemize}

\ta{Informally, if we consider an adversary \Adv and a $\transfer_i(v,j)$ invocation appearing in a global history $H$, where $p_i$ and $p_j$ are correct processes, the best \Adv can do (w.h.p.) is to pick $v \in \mathbb{R}^+$ and $j\in[1..n]$ uniformly at random.}
Doing so, \Adv will guess $j$ correctly with a probability inversely proportional to the number of active processes in $H$, and $v$ with a probability of $0$.\footnote{
    For simplicity, we assume that $v$ is defined on $\mathbb{R}^+$, but in practice, $v$ is a bounded positive number.
}
The above two properties capture that a Quasi-Anonymous Asset Transfer can be made arbitrarily close to this ideal situation by increasing the number of correct active processes and the security parameter $\lambda$ of the cryptographic schemes.
\ta{Furthermore, in our QAAT implementation (\Crefrange{alg:aat-init-vars}{alg:aat}), since each correct process $p_i$ is the only one knowing its initial account balance $\init_i$, QAAT-Confidentiality ensures that the account balance of $p_i$ also remains private throughout the execution with high probability.}



\section{A Modular Approach: Underlying Building Blocks} \label{sec:techniques}

The QAAT algorithm proposed in this paper follows a modular approach.
Before moving to our actual contributions in \Cref{sec:amt-algo}, we present in this section the cryptographic schemes upon which our QAAT system builds, namely, \textit{Commitments}, \textit{Universal Accumulators} (\textit{UA}), and \textit{Zero-Knowledge Proofs} (\textit{ZKP}).
Each scheme features a proving algorithm that aims to convince a verifying algorithm of some claim through a proof.
In the following, the proving algorithm is called the prover, while the verifying algorithm is called the verifier.
Due to space constraints, the full definition of the properties of these cryptographic schemes is given in~\Cref{sec:techniques-long}.

\mparagraph{Commitments.}
A \textit{commitment} scheme allows one to create a commitment $c$ to a chosen value $v$ while keeping this value hidden from others, with the ability to reveal the committed value later through a proof $o$ called opening.

A commitment scheme provides 2 operations: \textit{(i)} $\comCommit(v)/\langle c,o \rangle$ takes a value $v$ and outputs the corresponding commitment $c$ and opening $o$, and \textit{(ii)} $\comVerify(c,v,o)/r$ outputs $r=\ttrue$ if $o$ is a valid opening for the commitment $c$ and value $v$, and $r=\ffalse$ otherwise.
Informally, a commitment scheme must be \textit{binding} (it opens only to the committed value $v$) and \textit{hiding} (it does not leak information on its committed value).

\mparagraph{Universal Accumulators (UA).}
A \textit{universal accumulator} (\textit{UA}) scheme~\cite{BBF19} produces a short commitment $A$ (the accumulator) to a set of elements $E$, upon which the prover can produce both \textit{membership proofs} \mp or \textit{non-membership proofs} \nmp of elements $e$ in the set $E$, without divulging any other information.
We assume each process creates an  empty accumulator during the system's setup phase (see \Cref{sec:aat-setup}).

Consider an accumulator $A$, its associated accumulated set $E$, and an arbitrary value $v$.
A UA scheme provides 6 operations: \textit{(i)} $\accIsEmpty(A)/b$ takes $A$ and outputs $b=\ttrue$ if $E=\varnothing$, and $b=\ffalse$ otherwise; \textit{(ii)} $\accAdd(A,v)/A'$ takes $A$ and $v$ and outputs the updated accumulator $A'$ containing $v$; \textit{(iii)} $\accProveMem(A,E,v)/r$ takes $A$, $E$, and $v$ and outputs a membership proof $r=\mp$ of value $v$ in accumulator $A$ if $v \in E$, and $r=\abort$ otherwise; \textit{(iv)} $\accProveNonMem(A,E,v)/r$ takes $A$, $E$, and $v$ and outputs a non-membership proof $r=\nmp$ of value $v$ in $A$ if $v \notin E$, and $r=\abort$ otherwise; \textit{(v)} $\accVerifyMem(A,v,\mp)/r$ outputs $r=\ttrue$ if \mp is a correct membership proof of $v$ for $A$, and $r=\ffalse$ otherwise; and \textit{(vi)} $\accVerifyNonMem(A,v,\nmp)/r$ outputs $r=\ttrue$ if \nmp is a correct non-membership proof of $v$ for $A$, and $r=\ffalse$ otherwise.

Informally, a UA scheme must be \textit{sound} (proofs are unforgeable), \textit{complete} (the verification of \mbox{(non-)membership} works as intended), \textit{undeniable} (an element cannot have both membership and non-membership proofs), and \textit{indistinguishable} (the accumulator and proofs do not leak information on the set).

\mparagraph{Zero-Knowledge Proofs (ZKP).}
A \textit{zero-knowledge proof} (\textit{ZKP}) scheme~\cite{S20} produces proofs $\zkp$ that a prover knows some \secdata without divulging any other information on it to the verifier.
Each object $\mathit{Obj}_\mathit{ZK}$ of a ZKP scheme is set up with a specific predicate \Pzk, called a \textit{ZKP predicate}, taking as parameters  \pubdata (known both by the prover and verifier) and \secdata (known only by the prover), and returning \ttrue or \ffalse.
The prover $p_i$ passes to the \zkpprove operation the \pubdata and \secdata parameters of \Pzk, and the proof $\zkp_i$ is correctly generated only if $\Pzk(\pubdata,\secdata$$)=\ttrue$.
The \Pzk predicate is typically passed to the implicit setup operation of the ZKP scheme.

Consider a ZKP object $\mathit{Obj}_\mathit{ZK}$ set up with a ZKP predicate \Pzk.
A ZKP scheme provides 2 operations: \textit{(i)} $\mathit{Obj}_\mathit{ZK}.\zkpprove(\pubdata,\secdata)/r$ outputs a result $r$, which can be either a zero-knowledge proof $r=\zkp$ if the input parameters \secdata and \pubdata satisfy \Pzk or $r=\abort$ otherwise, and \textit{(ii)} $\mathit{Obj}_\mathit{ZK}.\zkpverify(\zkp,\pubdata)/r$ outputs the validity $r \in \{\ttrue,\ffalse\}$ of a zero-knowledge proof \zkp w.r.t. the \pubdata and the ZKP predicate \Pzk.
Informally, a ZKP scheme must be \textit{knowledge-sound} (if a proof \zkp is valid, then the prover knows the associated \secdata), \textit{complete} (a valid proof \zkp can be generated from some \pubdata and \secdata that satisfy \Pzk), and \textit{zero-knowledge} (a proof \zkp does not leak information on \secdata).

\section{A Light Consensus-Free Quasi-Anonymous Asset Transfer Algorithm} \label{sec:amt-algo}

Our novel asset-transfer system leverages the building blocks presented in \Cref{sec:techniques} and comprises two main contributions. The first, Agreement Proofs, presented in \Cref{sec:agreement-proof}  lies in a novel construct that produces transferable proofs that the processes have reached an agreement on some payload. The second, presented in \Cref{sec:algorithm} consists of our novel asset-transfer algorithm, which leverages the building blocks presented in \cref{sec:techniques}, and our novel Agreement Proofs to guarantee \emph{quasi-anonymity}, \emph{lightness}, and \emph{consensus freedom}.


\subsection{A new distributed scheme: Agreement Proofs (AP)} \label{sec:agreement-proof}

An \textit{Agreement proof} scheme (\textit{AP}) is a novel distributed scheme defined by two (explicit) operations $\approve$ and $\apverify$.
It aims at producing transferable agreement proofs (APs) that the system has reached an agreement regarding some payload value $v$ originating from a given sender/prover $p_i$ with sequence number $\sn_i$.
Sequence numbers uniquely identify each proof the prover generates using $\approve$, 
making the scheme \textit{multi-shot} (\ie a given prover can generate multiple proofs).
Our asset-transfer algorithm, presented in \Cref{sec:algorithm}, leverages APs to prevent double-spending.
Hence, the AP scheme, specified in the following, provides easily interpretable properties that formalize the notion of cryptographic certificates, which are often used in distributed systems~\cite{BSKD23}.
In \Cref{sec:ap-implem}, we provide a consensus-free AP implementation with a storage complexity of $O(\lambda+n)$ and a communication complexity of $O(n \lambda)$.

\mparagraph{AP predicate.}
Each object of an AP scheme is set up with a specific predicate \Pa, called an \textit{AP predicate}, taking in the AP's sequence number $\sn$ and some arbitrary $\data$ as parameters, and returning \ttrue or \ffalse.
The prover $p_i$ passes the payload $v$ and the \data parameter of \Pa to the \approve operation, which generates a correct proof $\ap_i$ only if $\Pa(v,\data,\sn_i)=\ttrue$, where $\sn_i$ is the sequence number of the current \approve invocation by $p_i$, \ie the number of times the prover $p_i$ has invoked \approve up to the current invocation.
We further assume that during system initialization, a valid \emph{genesis AP} can be generated by the set-up procedure for a pair $(v,\data)$ at sequence number $\sn_i=0$. (We return to the set-up procedure in \Cref{sec:init-var} and discuss implementation details in \Cref{sec:aat-setup}.)

\mparagraph{AP operations.}
We consider an AP object $\mathit{Obj}_A$ set up with an AP predicate \Pa.
\begin{itemize}[topsep=0pt,itemsep=-1ex,partopsep=1ex,parsep=1ex]
       \item $\mathit{Obj}_A.\approve(v,\data)/r$: Given the prover $p_i$, a payload value $v$ and some \data,  the operation returns $r=\ap_i$ if the predicate $\Pa(v,\data,\sn_i)$ is \ttrue, where $\sn_i$ is the sequence number of the current \approve invocation by $p_i$, and $\ap_i$ is an agreement proof for value $v$ at sequence number $\sn_i$, or $r=\abort$ otherwise.

    \item $\mathit{Obj}_A.\apverify(\ap_j,v,\sn_j,j)/r$: The operation returns the validity $r \in \{\ttrue,\ffalse\}$ of an agreement proof $\ap_j$ for a value $v$ of a prover $p_j$ at a sequence number $\sn_j$.
\end{itemize}

\mparagraph{Validity of an AP $\bm{\sigma}$.}
Given an AP object $\mathit{Obj}_A$, a payload value $v$, a sequence number $\sn_i$ and a prover $p_i$ (correct or faulty), we say that some \ap is a ``\textit{valid AP for $v$ at $\sn_i$ from $p_i$}'' if and only if any invocation of $\mathit{Obj}_A.\apverify(\ap,v,\sn_i,i)$ by any correct process would return true.

\mparagraph{AP properties.}
An AP object $\mathit{Obj}_A$ set up with an AP predicate \Pa satisfies four properties. 

\begin{itemize}[topsep=0pt,itemsep=-1ex,partopsep=1ex,parsep=1ex]
    \item \textbf{AP-Validity.} \sloppy 
    If $\ap_i$ is a valid AP for a value $v$ at sequence number $\sn_i > 0$ from a correct prover $p_i$, then $p_i$ has executed $\mathit{Obj}_A$.$\approve(v,\star)/\ap_i$ as its $\sn_i$\textsuperscript{th} invocation of $\mathit{Obj}_A$.$\approve(\star,\star)$.
    
    \item \textbf{AP-Agreement.}
    There are no two different valid APs $\ap_i$ and $\ap_i'$ for two different values $v$ and $v'$ at the same sequence number $\sn_i$ and from the same prover $p_i$.
    More formally, $\mathit{Obj}_A.\apverify(\ap_i,v,\sn_i,i)=\mathit{Obj}_A.\apverify(\ap_i',v',\sn_i,i)=\ttrue$ implies $v=v'$.

    
    


    \item \textbf{AP-Knowledge-Soundness.}
    If $\ap_i$ is a valid AP for value $v$ at sequence number $\sn_i > 0$ from a prover $p_i$ (correct or faulty), then
        $p_i$ knows some \data such that $\Pa(v,\data,\sn_i)/\ttrue$.

    \item \textbf{AP-Termination.}
    Given a correct process $p_i$ that executes $\mathit{Obj}_A.\approve(v,\data)/r$ with value $v$ and a \data, if $\Pa(v,\data,\sn_i)=\ttrue$ where $\sn_i$ is the sequence number of the current \approve invocation by $p_i$, then $r=\ap_i$ where $\ap_i$ is a valid AP for $v$ at $\sn_i$ from $p_i$.
    If $\Pa(v,\data,\sn_i)=\ffalse$ then $r=\abort$.
    
    

\end{itemize}

\subsection{QAAT Algorithm}\label{sec:algorithm}
Our quasi-anonymous asset-transfer (QAAT) algorithm leverages agreement proofs to guarantee that there can be at most one transfer per sequence number from any process $p_i$ ensuring that  a process cannot spend  the same funds twice.  In addition,  it leverages non-membership proofs to guarantee  that a given transfer is not already in the receiver's accumulator, ensuring that a process cannot redeem the same transfer twice. 
\subsubsection{Setup and initialization of process variables} \label{sec:init-var}

We present in \Cref{alg:aat-init-vars} the initialization of each of the variables maintained by each process $p_i$ in our system: $\bal_i$ (the balance of process $p_i$, only known to $p_i$ and initialized to $\init_i$), $\sn_i$ (the current sequence number of $p_i$, \ie the one for $p_i$'s latest transfer, initialized to 0), $\balcom_i$ and $\balopen_i$ (respectively the commitment and opening of $\bal_i=\init_i$), $T_i$ (the set of transfers details of $p_i$, including debits and credits, initially empty), $A_i$ (a universal accumulator of $T_i$, recording all $p_i$'s debits and credits and initialized to the empty accumulator), and $\ap_i$ (the previous agreement proof of $p_i$, initialized to a valid AP for the initial accumulator and balance commitment of $p_i$).
We assume that the $\bal_i$ and $\sn_i$ variables are of constant size (\eg 64 bits).

\begin{algorithm}[!tb]
\Init{%
$\bal_i \gets \init_i$;
$\sn_i \gets 0$;
$\aangle{\balcom_i, \balopen_i} \gets \comCommit(\init_i)$;
$T_i \gets \varnothing$;
\\
\nonl \hspace{2.1em}
$A_i$ $\gets$ empty universal accumulator of $p_i$; \\
\nonl \hspace{2.1em}
$\ap_i$ $\gets$ initial valid agreement proof for $\aangle{A_i,\balcom_i}$ at $\sn_i=0$ from $p_i$.
} \label{line:aat-init}
\caption{Initialization of the variables of \Cref{alg:aat} (code for $p_i$).}
\label{alg:aat-init-vars}
\end{algorithm}


To initialize all these variables, and in particular $\ap_i$, which involves communication among the processes, we employ the trustless setup procedure presented in \Cref{sec:aat-setup}.

\subsubsection{AP and ZKP predicates}
\IncMargin{.5em}

\begin{algorithm}[!tb]
\caption{AP and ZKP predicates of \Cref{alg:aat}.}
\label{alg:pred}
\setcounter{AlgoLine}{\getrefnumber{line:aat-init}}
\Init{%
\begin{tabular}[t]{@{}l@{ }c@{ }c@{ }l}
\AccountUpdate &$\gets$& AP & object setup with \Pa;\\
\TransferValidity &$\gets$& ZKP & object setup with \Pzk.
\end{tabular}
} \label{line:pred-init}
\AlgoSkip

\Predicate{$\Pa(\aangle{A'_\prover,\balcom'_\prover},\langle\ap_\prover,\zkp_\prover,\olddata_\prover\rangle,\sn_\prover)$\label{line:pa-data}%
}{%
    \SpacedCommentLine{Unpacking data about the prover's preceding state}
    $\aangle{A_\prover,\balcom_\prover} \gets \olddata_\prover$; \label{line:pa-olddata-unwrap} \\
    
    $\pubdata_\prover\gets \aangle{A_\prover,\balcom_\prover} \oplus\aangle{A'_\prover,\balcom'_\prover,\sn_\prover}$; \label{line:pa-pubdata-wrap}

    \SpacedCommentLine{$\zkp_\prover$ proves the knowledge of a valid transfer 
    fulfilling the ZKP predicate \Pzk (\cref{line:Pzk:start})
    from $\aangle{A_\prover,\balcom_\prover}$ to $\aangle{A'_\prover,\balcom'_\prover}$
    }
    \Assert $\TransferValidity.\zkpverify(\zkp_\prover,\pubdata_\prover)$;\label{line:pa-check-zkp}

    \SpacedCommentLine{If this is the prover's first transfer, its preceding  accumulator $A_\prover$ is empty}
    \lIf{$\sn_\prover = 1$}{
      \Assert $\accIsEmpty(A_\prover)$; 
    } \label{line:pa-check-first-trf}

    \SpacedCommentLine{The prover's preceding AP $\ap_\prover$ is valid at sequence number $\sn_\prover-1$}
    \Assert $\AccountUpdate.\apverify(\ap_\prover, \aangle{A_\prover,\balcom_\prover}, \sn_\prover-1, \prover)$. \label{line:pa-check-ap}
}
\AlgoSkip

\Predicate{$\Pzk(\pubdata,\secdata)$\label{line:Pzk:start}}{%
    \SpacedCommentLine{Unpacking public and private data}
    $\aangle{A_\prover, \balcom_\prover, A'_\prover, \balcom'_\prover, \sn_\prover} \gets \pubdata$; \label{line:pzk-pubdata-unwrap}
    
    {\setlength{\thinmuskip}{2mu}\setlength{\thickmuskip}{1mu}
    $\aangle{A'_\snd, \balcom_\snd, \ap_\snd, \mp_\snd, \tx, \bal_\prover, \balopen_\prover,$
    $ \balopen'_\prover, \nmp_\prover}\gets\secdata;\hspace{-5em}$\label{line:pzk-secdata-unwrap}}

    \SpacedCommentLine{The public commitment for balance should match the private data}
    \Assert $\comVerify$$(\balcom_\prover,\bal_\prover,\balopen_\prover)$;\label{line:pzk-check-com-bal}
    
    \label{line:pzk-check-com-trf}

    \SpacedCommentLine{The prover's accumulator has received the transfer $\tx$, $A'_\prover=A_\prover\uplus\{\tx\}$}
    \Assert $\accVerifyNonMem(A_\prover,\tx,\nmp_\prover) \land \accAdd(A_\prover,\tx) = A'_\prover$;\label{line:pzk-check-nonmem-trf}
    \label{line:pzk-check-new-acc}
    
    \SpacedCommentLine{The prover's balance has been properly updated}
    $\aangle{\snd, v, \rcv, \sn_\snd} \gets \tx$; \Assert $v \geq 0$;\label{line:pkz-start-check-balance}\label{line:pzk-trf-unwrap}

    \SpacedCommentLine{Trfs of sending APs are tagged with the AP's s.n. / a sending process has sufficient funds}
    \lIf{$\prover = \snd$}{\Assert $\sn_\snd = \sn_\prover \wedge \bal_\prover \geq v$;\label{line:pzk-sending-transfer-same-sn}}

    \lIf{$\prover=\snd=\rcv$}{\Assert $\comVerify$$(\balcom'_\prover, \bal_\prover, \balopen_\prover')$;} \label{line:pzk-check-snd-is-rcv}

    \lElseIf{$\prover = \snd$}{\Assert $\comVerify$$(\balcom'_\prover, \bal_\prover-v, \balopen'_\prover);
    \hspace{-3em}$}\label{line:pzk-check-pvr-is-snd}
    \ElseIf{$\prover = \rcv$}{%
        \Assert $\comVerify$$(\balcom'_\prover, \bal_\prover+v, \balopen'_\prover)$; \label{line:pkz-end-check-balance}
        
        \SpacedCommentLine{A receiving prover implies a valid corresp. sending AP $\ap_\snd$ from the sender}
        \Assert $\accVerifyMem$$(A'_\snd, \tx, \mp_\snd)$; \label{line:pzk-valid-sending-AP-exists-start}
        
        \Assert $\AccountUpdate.\apverify(\ap_\snd, \aangle{A'_\snd, \balcom_\snd},$
        $ \sn_\snd, \snd)$; \label{line:pzk-valid-sending-AP-exists-end}
    } 
    
    \lElse{\return \ffalse.} \label{line:pzk-ret-false}
} \label{line:pred-end}
\end{algorithm}

Our system relies on an AP object \AccountUpdate, and a ZKP object \TransferValidity, both shown in \Cref{alg:pred}. They are respectively set up with the AP predicate \Pa and the ZKP predicate \Pzk on \cref{line:pred-init}.
A process creates new AP and ZKP objects each time it sends or receives an asset transfer.
Hence, in these predicates, the prover can either be the sender or receiver of the asset transfer at hand.
We also assume that these predicates have access to the identity of the prover, denoted \prover. 
In the following, \snd and \rcv refer to the sender and receiver (resp.).
The statement \Assert $B$, where $B$ evaluates to a Boolean, is a shortcut for ``\textbf{if} $\neg B$ \textbf{then} \return \ffalse''.
Predicates return \ttrue by default.

\mparagraph{Predicate ${\bm{P_A}}$.}
The \Pa predicate takes 3 parameters: the payload $v$ (used both for the AP generation and verification), the \data (used only for the AP generation), and the current sequence number of the prover $\sn_\prover$.
The $v$ parameter contains the new accumulator of the prover, $A'_\prover$, as well as the commitment to its new balance, $\balcom'_\prover$.
The \data parameter contains the previous AP of the prover, $\ap_\prover$, the current ZKP of the prover, $\zkp_\prover$, and some data $\olddata_\prover$ about the state of the prover before the transfer.
This $\olddata_\prover$ contains the prover's preceding accumulator $A_\prover$ and a commitment to its previous balance $\balcom_\prover$.

\Pa starts with some unpacking (\cref{line:pa-olddata-unwrap}), and tuple concatenation (operation $\oplus$) to obtain the public data $\pubdata_\prover$ describing the current transfer (\cref{line:pa-pubdata-wrap}). $\pubdata_\prover$ encompasses the prover's accumulator before applying the transfer $A_\prover$, a commitment $\balcom_\prover$ to the prover's old balance, the prover's accumulator after the transfer $A_\prover'$, and a commitment $\balcom_\prover'$ to the prover's new balance.
\Pa then checks that the ZKP $\zkp_\prover$ is valid (\cref{line:pa-check-zkp}) using the public data describing the transfer.
Then, if this is the first transfer of the prover (debit or credit), \Pa checks that the prover's preceding accumulator is empty (\cref{line:pa-check-first-trf}). 
Finally, \Pa verifies the prover's preceding AP (\cref{line:pa-check-ap}).

\mparagraph{Predicate ${\bm{P_\mathit{ZK}}}$.}
The \Pzk predicate takes two parameters as input: \pubdata and \secdata. The former, \pubdata, consists of the tuple constructed at \cref{line:pa-pubdata-wrap}.
The latter, \secdata, consists of the sender's accumulator $A_\snd$ (recall that the sender is not always the prover), a commitment $\balcom_\snd$ of the sender's balance, the sender's AP $\ap_\snd$, the sender's membership proof $\mp_\snd$ that the transfer is in $A_\snd$ (notice that all the previous parameters are equal to $\bot$ if the prover is not the receiver as in this case there are not used in the body of \Pzk), the transfer details $\tx$, the prover's opening of the transfer $\balopen_\prover$, the prover's balance $\bal_\prover$, commitments of the prover's balance before and after applying the transfer $\bal_\prover$ and $\bal'_\prover$ (resp.), and the prover's non-membership proof $\nmp_\prover$ that the transfer was not already in $A_\prover$ (\cref{line:pzk-secdata-unwrap}).

All the checks of \Pzk are done in zero-knowledge, without divulging any data to the verifier(s).
\Pzk first checks that the commitments to the prover's balance and to the prover's transfer indeed open respectively to its balance (\cref{line:pzk-check-com-bal}).
\Pzk then checks that the prover's accumulator has been correctly updated, \ie that the transfer did not already belong to the prover's old accumulator, and that the prover's new accumulator can be obtained by adding the transfer to its old accumulator (\cref{line:pzk-check-nonmem-trf,line:pzk-check-new-acc}).
Finally, the \Pzk predicate verifies two properties of the transfer: \emph{(i)} that the prover's balance has been updated according to the transfer's information in $\tx$ (\crefrange{line:pkz-start-check-balance}{line:pkz-end-check-balance}), and \emph{(ii)}, in case the prover is the receiver, that the passed AP $\ap_\snd$ is valid and matches $\tx$ on the sender's side (\crefrange{line:pzk-valid-sending-AP-exists-start}{line:pzk-valid-sending-AP-exists-end}).
In more detail, \Pzk first verifies that if the prover is a sender, the transfer $\tx$ is stamped with the same sequence number $\sn_\prover$ as that of the prover's current AP (\cref{line:pzk-sending-transfer-same-sn}, $\sn_\prover$ is passed as a parameter to \Pzk from \Pa at \cref{line:pa-check-zkp}).
Then, if the prover, sender and receiver are the same, \Pzk checks that the prover's balance remains unchanged (\cref{line:pzk-check-snd-is-rcv})\footnote{
    This condition is needed to support empty transfers, which are used to enhance the sender anonymity of the system.
}.
Otherwise, if the prover is only the sender, \Pzk{} ensures the transfer amount has been subtracted from the prover's balance, and that the prover had enough funds to perform the transfer (\cref{line:pzk-check-pvr-is-snd}).
Finally, if the prover is only the receiver, then \Pzk verifies that the new prover's balance was obtained by adding the transfer's amount to its old balance (\cref{line:pkz-end-check-balance}), that the sender's accumulator contains the transfer (\cref{line:pzk-valid-sending-AP-exists-start}), and that the sender's AP is valid (\cref{line:pzk-valid-sending-AP-exists-end}).
If the prover is neither the sender nor the receiver of $\tx$, \Pzk returns \ffalse (\cref{line:pzk-ret-false}).

\subsubsection{Algorithm} \label{sec:qaat-algo}
\Cref{alg:aat} presents the code of our QAAT implementation for a process $p_i$.
The \balancee operation simply returns the value of $\bal_i$ (\cref{line:aat-blc}).
In the \transfer operation, $p_i$ first checks it has enough funds (\cref{line:aat-trf-check-bal}).
Then $p_i$ creates the transfer details $\tx$ with its next sequence number (\cref{line:aat-trf-create-trf}), processes its own transfer using the \prove internal operation (\cref{line:aat-trf-process}), creates a membership proof $\mp_i$ of the transfer in its accumulator (\cref{line:aat-trf-prove-mem}), sends all the necessary information to the receiver in a \transferm message (\cref{line:aat-trf-send}) and returns \commit (\cref{line:aat-trf-ret-commit}).
When $p_i$ receives a \transferm message, it processes the transfer using the \prove internal operation if it is the transfer receiver, and if the sender's AP and membership proof are valid (\cref{line:aat-rcv-trf-process}).

\IncMargin{.5em}

\begin{algorithm}[!tb]
\setcounter{AlgoLine}{\getrefnumber{line:pred-end}}
\newcommand{\llIf}[2]{{\let\par\relax\lIf{#1}{#2}}}

\lOperation{$\balancee()$}{\return $\bal_i$.} \label{line:aat-blc}
\AlgoSkip

\Operation{$\transfer(v,j)$}{
    \lIf{$v<0 \;\lor\; v > \bal_i$}{\return \abort;} \label{line:aat-trf-check-bal}
    
    $\tx \gets \aangle{i,v,j,\sn_i+1}$; \label{line:aat-trf-create-trf} \\
    
    $\prove(\tx,\bot,\bot,\bot,\bot)$; \label{line:aat-trf-process}
    \Comment*{Producing ZKP and Agreement Proof for $\tx$}
    
    $\mp_i \gets \accProveMem(A_i,T_i,\tx)$; \label{line:aat-trf-prove-mem} \Comment*{Membership proof that $\tau$ is now in $A_i$}
    
    \rasend $\transferm(\tx,\ap_i,A_i,\mp_i,\balcom_i)$ to $p_j$; \label{line:aat-trf-send} \Comment*{\rasend is receiver-anonymous}
    
    \return \commit. \label{line:aat-trf-ret-commit}
}
\AlgoSkip

\When{$\textup{\transferm}(\tx,\ap_j,A_j,\mp_j,\balcom_j)$ {\bf is} \rareceived\ {\bf from} $p_j$}{%
    \If{$\left\{\begin{array}{@{}l@{}l@{}}\text{$\tx$.$\rcv$}=i \land \accVerifyMem(A_j,\tx,\mp_j) \land\\ \AccountUpdate.\apverify(\ap_j, \aangle{A_j,\balcom_j}, \tx.\sn, j)\\ \end{array}\right\}$}{ \label{line:aat-rcv-trf-cond}
        $\prove(\tx,A_j,\balcom_j,\ap_j,\mp_j)$. \label{line:aat-rcv-trf-process}
    }
}
\AlgoSkip

\InternalOperation{$\prove(\tx, A_j,\balcom_j,\ap_j,\mp_j)$}{\label{line:aat-prove-start}
    
    \SpacedCommentLine{Computing new balance with corresponding commitment and opening}
    $\bal'_i \gets \bal_i$;\label{line:aat-process-update-bal-start}
    \llIf{$i = \tx.\snd$}{$\bal'_i \gets \bal'_i-\tx.v$} ; 
    \lIf{$i = \tx.\rcv$}{$\bal'_i \gets \bal'_i+\tx.v$}\label{line:aat-process-update-bal} \label{line:aat-process-update-bal-end}
    
    $\aangle{\balcom'_i,\balopen'_i} \gets \comCommit$$(\bal'_i)$; \label{line:aat-process-com-bal} \Comment*{New commitment and opening for $\bal'_i$}

    $\nmp_i \gets \accProveNonMem(A_i,T_i,\tx)$; \label{line:aat-process-prove-nonmem}\Comment*{Non-membership proof, $\tau$ not in $A_i$}

    $A'_i \gets \accAdd(A_i,\tx)$; \label{line:aat-zkp-start}\label{line:aat-process-add-trf} \Comment*{Adding new transfer to accumulator}
    \SpacedCommentLine{Constructing ZK proof that transfer is valid}
    \ft{$\olddata \gets \aangle{A_i, \balcom_i}$}; \label{line:aat-process-create-olddata} \\

    $\pubdata \gets \ft{\olddata \oplus \aangle{A'_i, \balcom'_i, \sn_i+1}}$; \label{line:aat-process-create-pubdata} \\
    
    $\secdata \gets \aangle{A_j, \balcom_j, \ap_j,\mp_j, \tx, \bal_i, \balopen_i, \balopen'_i, \nmp_i}$; \label{line:aat-process-create-secdata} \\
    
    $\zkp \gets \TransferValidity.\zkpprove(\pubdata, \secdata)$; \label{line:aat-process-create-zkp}\label{line:aat-zkp-end} \\
    

    \SpacedCommentLine{Obtaining Agreement Proof (AP) on transfer's validity}

    $\data \gets \aangle{\ap_i, \zkp, \ft{\olddata}}$; \label{line:aat-process-create-data}
    
    $\ap_i \gets \AccountUpdate.\approve(\aangle{A'_i,\balcom'_i},\data)$; \label{line:aat-process-create-ap}
    \\

    \SpacedCommentLine{Updating local state}
    
    $\sn_i \gets \sn_i+1$; $A_i \gets A'_i$; $\bal_i \gets \bal'_i$; $\balcom_i \gets \balcom'_i$;
    $\balopen_i \gets \balopen'_i$. \label{line:aat-process-update-vars}
    \label{line:aat-prove-end}
}

\caption{QAAT algorithm (code for $p_i$).}
\label{alg:aat}
\end{algorithm}
\TA{We could simplify the algorithm even further by removing the \accVerifyMem operation: we only need the old accumulator and the new accumulator (after the \accAdd) to verify the membership.}

In the \prove internal operation, $p_i$ first 
computes its new balance $\bal'_i$ depending on whether it is the sender or receiver (\cref{line:aat-process-update-bal-start}),
and commits its new value (\cref{line:aat-process-com-bal}).
Next, $p_i$ proves the non-membership of the transfer $\tx$ in its old accumulator $A_i$ (\cref{line:aat-process-prove-nonmem}) and creates its new accumulator by adding the transfer (\cref{line:aat-process-add-trf}).
Process $p_i$ then constructs the ZK proof $\zkp$ that the transfer $\tx$ is valid 
(\crefrange{line:aat-process-create-olddata}{line:aat-zkp-end}). 
It then creates the public data (\cref{line:aat-process-create-pubdata}) and secret data (\cref{line:aat-process-create-secdata}) of the ZKP, and generates the ZKP (\cref{line:aat-process-create-zkp}).
Next, $p_i$ creates the data of the Agreement Proof (AP) predicate (\cref{line:aat-process-create-data}).
Finally, $p_i$ generates the AP $\ap_i$ (\cref{line:aat-process-create-ap}) before updating its local state (\cref{line:aat-process-update-vars}).


\subsubsection{Intuition of Algorithm~\ref{alg:aat}'s proofs} \label{sec:intuition-proofs-qaat}

In this section, we sketch the high-level intuition behind the correctness of our solution.
The full correctness developments are given in \Cref{sec:proofs-sys}.

\mparagraph{Correctness proof.}
The correctness of our system comes from the fact that it satisfies AT-Sequentiality and AT-Termination.

\begin{restatable}[AT-Sequentiality]{rlemma}{atsequentiality}
\label{lemma:at-sequentiality}
For any global history $H$ capturing an execution of \emph{\Cref{alg:aat}}, there exists a mock history $\mock{H}$ of $H$ that can be AT-sequenced.
\end{restatable}

\begin{proofsketch}
The proof first constructs the mock history $\mock{H}$ by starting from the UAs linked to the APs generated by correct processes, and then recursively traversing these UAs and APs using the predicates $\Pa$ and $\Pzk$ to uncover UAs issued by Byzantine processes that are causally linked to the operations of correct processes. Each UA yields information on a transfer invocation at a given sequence number, which we then use to construct the mock local execution $\mock{L}_j$ of each Byzantine process $p_j$.

Given a correct process $p_i$, we then construct a partial $\POATSeq$ order on $\mathcal{S}_i=\supportSet(\mock{L}_i)\cup\bigcup_{j\neq i}\transferSet(\mock{L}_j)$. This construction is somewhat technical for two reasons: (1) Transfers might not be received in the order they were issued (in particular by $p_i$), which implies $\POATSeq$ should not enforce any local process order $\rightarrow_j$ other than that of $p_i$. (2) Simultaneously, $\POATSeq$ should be constraining enough to guarantee that \balancee invocations by $p_i$ exactly reflect previous transfers (Property n.\ref{at-seq:balance:inv} of an AT-sequence in \Cref{sec:amt-spec}) and that each transfer is backed by sufficient funds (Property n.\ref{at-seq:positive:account}). We construct $\POATSeq$ incrementally using two binary relations: $\POATSeqOne$ that captures $p_i$'s local order and ensures that all transfers are sufficiently funded, and $\POATSeqTwo$ to guarantee \balancee invocations by $p_i$ can be properly explained. Part of the proof focuses on the acyclicity of $\POATSeqOne\cup\POATSeqTwo$, so that $\POATSeq$ can be defined as the transitive closure $(\POATSeqOne\cup\POATSeqTwo)^{+}$. Finally, we chose $\sequence_i$ as a topological sort of $(\mathcal{S}_i,\POATSeq)$ and prove it fulfills the two properties of AT-Sequences.
(Full derivations in \Cref{sec:proof-at}.)
\end{proofsketch}

\begin{restatable}[AT-Termination]{rlemma}{attermination}
\label{lemma:at-termination}
All operation invocations of \emph{\Cref{alg:aat}} (\balancee and \transfer) terminate for correct processes. \proofATinappendix
\end{restatable}

\TA{We have the space to move back the proof here, it is very short}

\renewcommand{\proofATinappendix}{}

\mparagraph{Quasi-anonymity proof.}
Intuitively, our system is quasi-anonymous because it uses cryptographic schemes that do not leak sensitive data (\ie commitments, universal accumulators, and zero-knowledge proofs).
We prove the following lemmas in \Cref{sec:proof-qaat}.


\begin{restatable}[QAAT-Receiver-Anonymity]{rlemma}{qaatrcvanonymity}
\label{lem:qaat-rcv-anonymity}
For any global history $H$ capturing an execution of $\mathcal{A}$, for any invocation $\transfer_i(\star,j)$ contained in $H$ where $p_i$ and $p_j$ are correct processes, then, with high probability, $p_j$ is indistinguishable among all the correct active processes of $H$.
\proofQAATinappendix
\end{restatable}


\begin{restatable}[QAAT-Confidentiality]{rlemma}{qaatconfidentiality}
\label{lem:qaat-confidentiality}
For any global history $H$ capturing an execution of $\mathcal{A}$, for any invocation $\transfer_i(v,j)$ contained in $H$ where $p_i$ and $p_j$ are correct processes, then, with high probability, $v$ is indistinguishable among any value in $\mathbb{R}^+$.
\proofQAATinappendix
\end{restatable}

\renewcommand{\proofQAATinappendix}{}

\mparagraph{Succinctness.}
The overall succinctness of our system stems from the succinctness of our Agreement Proof implementation (\Cref{sec:ap-implem-succinctness}), the constant size of the digests and proofs of RSA accumulators (\Cref{sec:acc-impl}) and of commitments (\Cref{sec:com-impl}), and of the succinctness of Spartan zk-SNARKs (\Cref{sec:zkp-impl}).
\ar{By definition, the proofs of computation produced by the prover of a succinct ZKP scheme are succinctly verified. Therefore, primitives whose verification is encapsulated inside the ZKP predicate only need to be constant size for our scheme to be succinct (although accumulator (non-)membership proof verification and commitment openings can, in fact, be succinct).}
\TA{I don't fully understand the previous sentences}

\mparagraph{Storage of $\bm{O(\lambda+(|T|/n)\log n + n)}$ bits per correct process.}
In our system, each correct process $p_i$ only stores its local transfer details, the sequence number of each of the $n$ processes, and a few cryptographic structures, yielding a storage of $O(\lambda+(|T|/n) \log n + n)$, where $T$ is the set of all transfers in the system.
Let us consider the size of $p_i$'s local variables.
By definition, $\bal_i$ and $\sn_i$ are constant-size (see \Cref{sec:qaat-algo}).
Moreover, $\balcom_i$ and $\balopen_i$ have $O(\lambda)$ bits with a constant size commitment scheme (\Cref{sec:com-impl}), $A_i$ has $O(\lambda)$ bits with the constant size RSA accumulator implementation (\Cref{sec:acc-impl}), and $\ap_i$ has $O(\lambda)$ bits with our implementation of agreement proofs (\Cref{sec:ap-implem}).
Moreover, our AP algorithm stores $O(\lambda+n)$ bits per correct process (\Cref{sec:ap-implem}). 
Finally, the $T_i$ set is of size $O(|T|/n)$ (where $T$ is the set of all transfers of the system) and contains transfer details of size $O(\log n)$ bits (because of process IDs).
This amounts to a total storage cost of $O(\lambda+(|T|/n)\log n + n)$.

\mparagraph{Communication of $\bm{O(n \lambda)}$ bits overall.}
In our system, a $\transfer_i(\star,j)$ invocation by a correct process $p_i$ entails the following communications: \emph{(i)} the $\approve$ invocation by the sender $p_i$ at \cref{line:aat-process-create-ap}, \emph{(ii)} the $\rasend$ of the transfer details from the sender to the receiver at \cref{line:aat-trf-send}, and \emph{(iii)} the $\approve$ invocation by the receiver $p_j$ at \cref{line:aat-process-create-ap} (which is only guaranteed to happen if $p_j$ is correct).
As our Agreement Proof implementation (\Cref{sec:ap-implem}) communicates only $O(\lambda n)$ bits overall (\Cref{sec:ap-implem-comm}), \emph{(i)} and \emph{(iii)} incur an overall communication of $O(\lambda n)$.
For \emph{(ii)}, the receiver-anonymous $\rasend$ operation can be implemented with an overall communication cost of $O(\lambda n)$, as discussed in \Cref{sec:model}.
As a result, our system's overall communication cost is $O(\lambda n)$.

\mparagraph{Consensus-freedom.}
This directly comes from our distributed Agreement Proof implementation (\Cref{sec:ap-implem}), which does not rely on strong agreement such as consensus (\Cref{sec:ap-implem-cons-free}). 

\subsubsection{Further enhancements} \label{sec:enhancements}

\mparagraph{Transfer batching.}
One could argue that requiring users to commit an accumulator update to the system each time they send or receive a single transfer is inefficient. To address this problem, we propose aggregating an arbitrary number of transfers into a single accumulator update.
As a result, a user could choose only to declare an accumulator update when making a payment, while cashing all receipts simultaneously when doing so.
To implement transfer batching, we use an optimization method for ZK proofs known as \emph{folding} \cite{KST22} in \Cref{sec:batching}.

\mparagraph{Reducing local storage (public key rotation).}
The QAAT system presented in this article is light, \ie the storage cost per process $p_i$ is in $O(\lambda+(|T|/n)\log n + n)$, where $|T|$ is the set of all transfers (debits and credits) of the system.
This storage cost is justified by $p_i$'s need to store some data whose size is proportional to its entire transfer history, to prove that it is not trying to redeem the same transfer several times.
However, a public key rotation mechanism could be added to remove the need to store past transfers.
When a process $p_i$ rotates its public key, it changes its old public key for a new one, while transferring all its funds to the account associated to the new public key.
Once this is done, $p_i$ can flush its old local data (and, in particular, its accumulator data).
Thus, the process only has to record information concerning this rotating transfer which can be seen as the genesis state of a newly created account, and obtains a storage cost of $O(\lambda + n)$.
However, this mechanism would involve one major technical challenge: a sender has to retrieve the receiver's public key before it can send funds to her.
To achieve this retrieval, the sender can initiate a handshake with the receiver, but due to asynchrony and the fact that the receiver may be faulty, the handshake may never complete, hampering the termination of the transfer.
This method is briefly mentioned by Zef~\cite{BSKD23} as a solution to safely delete the data of unused accounts. To circumvent the incompatibility of its asynchronous model and the handshake, Zef assumes that each transfer is initiated in a synchronous environment outside the system where the receiver transmits its public key to the sender.


\mparagraph{Towards full anonymity.} Our system is not fully anonymous, as it guarantees Receiver Anonymity and Confidentiality, but not that the transfer issuer remains anonymous (\ie Sender Anonymity).
However, remark that our system allows ``empty'' transfers with amount 0 (or transfers to oneself), which do not change any balance.
These empty transfers could be used to obfuscate the traffic of funds from the adversary's point of view, and if they are issued at the right moment, we conjecture that they could preclude (w.h.p.) timing attacks, \ie attacks where an adversary deanonymizes the sender (or receiver) of an asset transfer by observing the timing of messages transiting on the network.
If so, our system would reach full anonymity asymptotically (by adding more empty transfers at some well-chosen moments).
The design of the heuristics to choose the moments to issue empty transfers is left to future work.


\section{Conclusion} \label{sec:conclusion}
This article considered the problem of asynchronous Byzantine-fault-tolerant asset transfer, with the additional constraint of satisfying the properties of quasi-anonymity (\ie no leak of information on the transfers' amounts and receivers), lightness (\ie succinct cryptography and light storage cost), and consensus-freedom (\ie no total order of transfers).
These properties are important for achieving good performance, confidentiality, and user privacy in an asset transfer system.
In this context, this article introduced Quasi-Anonymous Asset Transfer (QAAT), and presented a consensus-free and light QAAT implementation, along with its formal proofs.
To our knowledge, our asset transfer system is not only the first to fulfill the 3 properties, but also the first one to have a $O(\lambda |T|/n)$ storage cost, where $T$ is the set of all transfers of the system.
In addition, the article presented a new distributed abstraction called Agreement Proofs, which captures the notion of distributed agreement in a transferable short-sized proof.

Presently, our asset transfer solution still lacks some capabilities compared to more mature blockchain systems (\eg Sybil resistance or smart contracts) or mainstream payment networks (\eg overdrawn accounts or fraud detection), but we believe that it demonstrates that systems with low verification, storage, and network costs can still guarantee strong privacy features.
Furthermore, we conjecture that our system's scalability can be further enhanced, and in particular, that it can be made permissionless without sacrificing consensus-freedom, by leveraging techniques such as the ones introduced in~\cite{KPPT23}.


\newpage

\newpage

\appendix

\section{Underlying Building Blocks: Full Definitions} \label{sec:techniques-long}

This appendix section presents the full definitions of the cryptographic schemes used in our QAAT system, namely Commitments, Universal Accumulators (UA), and Zero-Knowledge Proofs (ZKP).
More precisely, this section provides the operations and properties of these schemes.

These schemes all require an implicit \textit{setup} operation that takes as input the security parameter $\lambda$ and outputs the public parameters of the scheme, which are known to all processes, including the adversary.
For instance, in a digital signature scheme, public parameters correspond to the set of all the system's public keys.
All the cryptographic schemes that follow implicitly use some public parameters (although these parameters are not exposed explicitly in the specifications and algorithms).
The domains of definition for the input values and output proofs of our schemes are respectively called $\mathbb{D}$ and $\mathbb{F}$ (see \Cref{tab:crypto-in-out-sets}).
Each scheme features a proving algorithm that aims to convince a verifying algorithm of some claim through a proof.
In the following, the proving algorithm is called the prover, while the verifying algorithm is called the verifier.
Formally, both the prover and the verifier are poly-time probabilistic algorithms.
For simplicity, the schemes presented in this section do not have explicit termination properties, as their operations only consist of local mathematical computations.

We further denote by $\epsilon(\lambda)$ a positive number between $0$ and $1$ that can be made arbitrarily small by increasing $\lambda$.

\begin{table*}[hb]
    \centering
    \begin{tabular}{|c|c|}
    \hline
    \textbf{Set} & \textbf{Meaning} \\
    \hhline{|=|=|}   
    \hline
    $\mathbb{D}$ (for ``data'') & Input set (values) of the proving operations of the cryptographic schemes \\
    \hline
    $\mathbb{F}$ (for ``finite'') & Output set (proofs) of the proving operations of the cryptographic schemes \\
    \hline
    \end{tabular}
    \vspace{-.2em}
    \caption{Input and output sets of the cryptographic schemes of \Cref{sec:techniques-long}.} \label{tab:crypto-in-out-sets}
\end{table*}



\subsection{Commitment schemes} \label{sec:commitment}
A commitment scheme is a cryptographic primitive that enables a prover to succinctly commit to a value $v$ in the form a small digest $c$ (the commitment) and compute a proof $o$ (the opening) that the commitment ``opens'' to $v$.
Commitment schemes must be binding, \ie the prover cannot find another value that matches the digest.
Most schemes also have a hiding property, meaning that the digest does not reveal any information about the committed value.
Some schemes allow digests to be combined to perform operations (\eg addition or multiplication) on committed values, such as homomorphic schemes.
There are many variants of commitments to handle different types of data, such as polynomials, vectors, or even functions.
In our system, we use a simple \textit{commitment} scheme (\textit{C}) to a bounded integer.
%

\mparagraph{Commitment operations.}
\begin{itemize}
    
    \item $\comCommit(v)/\langle c,o \rangle$ takes a value $v \in \mathbb{D}$ and outputs the corresponding commitment $c$ and opening $o$.
    
    \item $\comVerify(c,v,o)/r$ outputs $r=\ttrue$ if $o$ is a valid opening for the commitment $c$ and value $v \in \mathbb{D}$, and $r=\ffalse$ otherwise.
\end{itemize}

\mparagraph{Commitment properties.}

\begin{itemize}
    \item \textbf{C-Correctness.} $\Forall v \in \mathbb{D}$, if $\comCommit(v)/\aangle{c,o}$ then $\comVerify$$(c,v,o)/\ttrue$.

    Informally, a commitment created from a value opens to this value (using an opening).
    
    \item \textbf{C-Binding.}
    $\Forall c, o, o' \in \mathbb{F}, \Forall v, v' \in \mathbb{D}: \Pr\big(\comVerify$$(c,v,o)/\ttrue \land \comVerify(c,v',o')/\ttrue \land v' \neq v\big) < \epsilon(\lambda).$
    Informally, a commitment opens to only one value (w.h.p.).
    
    \item \textbf{C-Hiding.}
    $\Forall c_1,c_2 \in \mathbb{F}, \Forall v \in \mathbb{D}, \comCommit(v)/\aangle{c,\star}: |\Pr(c=c_1) - \Pr(c=c_2)| \!<\! \epsilon(\lambda)$.

    Informally, a commitment does not leak any information on its committed value (w.h.p).
\end{itemize}

\subsection{Universal Accumulators (UA)} \label{sec:accumulator} 
An \textit{accumulator} scheme (notion introduced in~\cite{BM93}) produces a short commitment to a set of elements $E$.
A \textit{universal accumulator} (\textit{UA}) is a special kind of accumulator that can prove both the inclusion or non-inclusion of an element in the set by using \textit{membership} or \textit{non-membership proofs}, respectively.
RSA accumulators~\cite{BBF19} are a notable implementation of this scheme.
We assume that the empty accumulator of each process is created during the system's setup phase (see \Cref{sec:aat-setup}).


\mparagraph{UA operations.} We consider an accumulator $A$ and its associated accumulated set $E$.

\begin{itemize}
    \item $\accIsEmpty(A)/b$: Takes $A$ and outputs $b=\ttrue$ if $E=\varnothing$, and $b=\ffalse$ otherwise.
    
    \item $\accAdd(A,v)/A'$: Takes $A$ and a value $v$ and outputs the updated accumulator $A'$ containing $v$.


    \item $\accProveMem(A,E,v)/r$: Takes $A$, $E$, and a value $v$ and outputs a membership proof $r=\mp$ of value $v$ in accumulator $A$ if $v \in E$, and $r=\abort$ otherwise.
    
    \item $\accProveNonMem(A,E,v)/r$: Takes $A$, $E$, and a value $v$ and outputs a non-membership proof $r=\nmp$ of value $v$ in $A$ if $v \notin E$, and $r=\abort$ otherwise.

    \item $\accVerifyMem(A,v,\mp)/r$: Outputs $r=\ttrue$ if \mp is a correct membership proof of value $v$ for $A$, and $r=\ffalse$ otherwise.
    
    \item $\accVerifyNonMem(A,v,\nmp)/r$: Outputs $r=\ttrue$ if \nmp is a correct non-membership proof of value $v$ for $A$, and $r=\ffalse$ otherwise.
\end{itemize}

\mparagraph{UA properties.} We consider an accumulator $A$ and its associated set $E$.

\begin{itemize}
    
    
    
    \item \textbf{UA-Soundness.} \sloppy $\Forall v \notin E, \accProveMem(A,E,v)/r: \Pr(r \neq \abort) < \epsilon(\lambda)$ and 
    $\Forall v \in E, \accProveNonMem(A,E,v)/r: \Pr(r \neq \abort) < \epsilon(\lambda)$.

    Informally, the probability of computing a membership proof for a non-accumulated element or a non-membership proof for an accumulated element is negligible.

    \item \textbf{UA-Completeness.} \sloppy $\Forall v \in E, \accVerifyMem(A,v,\accProveMem(A,E,v))/r: \Pr(r=\ttrue)>1-\epsilon(\lambda)$ and 
    $\Forall v \notin E, \accVerifyNonMem(A,v,\accProveNonMem(A,E,v))/r\!: \Pr(r=\ttrue)>1-\epsilon(\lambda)$.

    Informally, all honestly accumulated values are verified as true with their corresponding membership proof with a negligible probability of error.
    
    \item \textbf{UA-Undeniability.} $\Forall v \in \mathbb{D}$, $\Forall \mp,\nmp \in \mathbb{F}: \\ \Pr(\accVerifyMem(A,v,\mp) \land \accVerifyNonMem(A,v,\nmp))<\epsilon(\lambda)$.

    Informally, the probability of computing both a membership and non-membership proof for the same element is negligible.
    
    \item \textbf{UA-Indistinguishability.} No information on some accumulated set $E$ is leaked from its associated accumulator $A$, membership proofs $w$ or non-membership proofs $u$.%
    \footnote{For a more formal definition of Indistinguishability, we refer the interested reader to~\cite{DHS15}.}
\end{itemize}

For simplicity and compliance with the traditional definitions given in the cryptography literature~\cite{BBF19}, we do not give properties for the \accAdd operations.
We implicitly assume that they behave correctly, that is, we can only prove the membership of an element that has been added and we can only prove the non-membership for an element that has never been added.

\subsection{Zero-Knowledge Proofs (ZKP)} \label{sec:zkp}
We refer to the last scheme used in this article as \textit{zero-knowledge proofs} (\textit{ZKP}), which correspond to proofs that a prover knows some secret data without divulging any other information on it to the verifier.
Specifically, we use zero-knowledge Succinct Non-Interactive Arguments of Knowledge (zk-SNARKs)~\cite{S20}.
The difference between proof and argument systems comes from the strength of the soundness property.
A proof system must be perfectly sound (\ie withstand a computationally unbounded adversary) whereas an argument system guarantees computational soundness (\ie against a PPT adversary with a very high probability).
The use of arguments is necessary as it was proven by Fortnow~\cite{F89} that no perfect soundness zero-knowledge proof systems exist for NP-complete problems, while there exist perfect zero-knowledge arguments systems for NP-complete problems (\eg zk-SNARKs).

\mparagraph{ZKP predicate.}
\sloppy
Each object of a ZKP scheme is set up with a specific predicate \Pzk, called a \textit{ZKP predicate}, taking as parameters a \pubdata (known both by the prover and verifier) and a \secdata (known only by the prover), and returning \ttrue or \ffalse.
The prover $p_i$ passes to the \zkpprove operation the \pubdata and \secdata parameters of \Pzk, and the proof $\zkp_i$ is correctly generated only if $\Pzk(\pubdata,\secdata)=\ttrue$.
Much like the \Pa predicate (\Cref{sec:agreement-proof}), the \Pzk predicate is typically passed to the implicit setup operation of the ZKP scheme.

\mparagraph{ZKP operations.}
We consider a ZKP object $\mathit{Obj}_\mathit{ZK}$ set up with a ZKP predicate \Pzk.
\begin{itemize}
    \item $\mathit{Obj}_\mathit{ZK}.\zkpprove(\pubdata,\secdata)/r$: Returns a result $r$, which can be either a zero-knowledge proof $r=\zkp$ if the input parameters \secdata and \pubdata satisfy \Pzk, or $r=\abort$ otherwise.

    \item $\mathit{Obj}_\mathit{ZK}.\zkpverify(\zkp,\pubdata)/r$: Returns the validity $r \in \{\ttrue,\ffalse\}$ of a zero-knowledge proof \zkp with respect to the public data \pubdata and the ZKP predicate.
\end{itemize}


\mparagraph{ZKP properties.}
We consider a ZKP object $\mathit{Obj}_\mathit{ZK}$ set up with a ZKP predicate \Pzk.

\begin{itemize}
    \item \textbf{ZKP-Knowledge-Soundness.}
    If $\mathit{Obj}_\mathit{ZK}.\zkpverify(\zkp,\pubdata)$ is true, then the prover knows some \secdata such that $\mathit{Obj}_\mathit{ZK}.\zkpprove(\pubdata,\secdata)/\zkp$ and $\Pzk(\pubdata,\secdata)$ is true w.h.p.

    


    \item \textbf{ZKP-Completeness.}
    For any pair $\langle \pubdata,\secdata \rangle$ such that $\Pzk(\pubdata,\secdata)$ is true, we must have $\mathit{Obj}_\mathit{ZK}.\zkpverify(\mathit{Obj}_\mathit{ZK}.\zkpprove(\pubdata,\secdata),\pubdata)$ is true w.h.p.

    \item \textbf{ZKP-Zero-Knowledge.} No information on some \secdata is leaked by its associated \pubdata and zero-knowledge proof \zkp.\footnote{For a more formal definition of Zero-Knowledge, we refer the interested reader to~\cite{G001}.}

\end{itemize}

\section{Proofs of our QAAT System~(Algorithms~\ref{alg:aat-init-vars} to~\ref{alg:aat})} \label{sec:proofs-sys}

In this appendix section, we provide full derivations on the proof of correctness (\Cref{sec:proof-at}) and quasi-anonymity (\Cref{sec:proof-qaat}) of our QAAT system (\Crefrange{alg:aat-init-vars}{alg:aat}).
Throughout the section, we rely on the full property definitions of the commitment, UA, and ZKP schemes given in \Cref{sec:techniques-long}.

\subsection{Proof of correctness} \label{sec:proof-at}

\subsubsection{Preliminaries} \label{sec:proof-prelim}
We first introduced a number of preliminary results and definitions that we will use to prove the AT-Sequentiality of \Cref{alg:aat} in \Cref{sec:proof:at-seq}. 
In the following, $\ap_i^\sn$ denotes an AP (Agreement Proof) that is valid at a sequence number $\sn$ for a process $p_i$ (correct or faulty).

\begin{definition}[Valid UA at a sequence number $\sn$ for a process $p_i$]\label{def:valid:UA}
We say that a Universal Accumulator (UA) $A$ is valid at a sequence number $\sn$ for a process $p_i$ (correct of faulty) if there exists a valid Agreement Proof (AP) $\ap_i^\sn$ that is valid at $\sn$ for $p_i$ and verifies $\AccountUpdate.\apverify(\ap_i^\sn, \aangle{A,\star}, \sn, i)$. 
As a shortcut, we might interchangeably say that $A$ was \emph{issued} by $p_i$ at sequence number $\sn$.
When this holds, we note $A$'s sequence number and AP issuer as superscript and subscript, respectively, \ie $A=A_i^\sn$.
\end{definition}

\begin{lemma}[Unicity of valid UAs at a sequence number $\sn$ for a process $p_i$]\label{lemma:unicity:UAs}
    There is at most one UA that is valid at a sequence number $\sn\geq 0$ for a process $p_i$ (correct of faulty), \ie if $A$ and $A'$ are both valid UAs at $\sn$ for $p_i$ according to \Cref{def:valid:UA}, then $A=A'$.
\end{lemma}
\begin{proof}
    This trivially follows from \Cref{def:valid:UA} and the AP-Agreement property of Agreement Proofs.
\end{proof}

\begin{lemma}[Sequence of preceding UAs of a valid UA]\label{lemma:preceding:UAs}
\label{rem:seq:previous:UAs}
Consider a valid UA $A_i^\sn$ at sequence number $\sn\geq 0$ issued by a process $p_i$ (correct or faulty). The following holds.
\begin{itemize}
    \item $A_i^\sn$ can be associated with a unique \textit{sequence of preceding UAs}, denoted $\precedingUAs(A_i^\sn)=(A_i^0,A_i^1,..,A_i^\sn)$, so that each $A_i^k$ is a valid UA issued by $p_i$ at sequence number $k$, with $0 \leq k \leq \sn$.
    
    \item The sequences of UAs produced for UAs issued by the same process $p_i$ are prefix-ordered, \ie if $A_i^{\sn_1}$ and $A_i^{\sn_2}$ are two valid UAs issued by $p_i$ at sequence numbers $\sn_1$ and $\sn_2$ respectively such that $\sn_1\leq\sn_2$, then $\precedingUAs(A_i^{\sn_1})$ is a prefix of $\precedingUAs(A_i^{\sn_2})$.
\end{itemize}
\end{lemma}

\begin{proof}
  Let us consider a valid UA $A_i^\sn$ at sequence number $\sn\geq 0$ \ft issued by a process $p_i$ (correct or faulty).
  \begin{itemize}
      \item If $\sn=0$, we pick $\precedingUAs(A_i^0)=(A_i^0)$.
      \item If $\sn>0$, by \Cref{def:valid:UA}, $A_i^\sn$ is associated with some AP $\ap_i^\sn$ issued by $p_i$ that is valid at sequence number $\sn$.
      By AP-Knowledge-Soundness, $p_i$ knows some $\data=\aangle{\ap_i',\star,\star}$ s.t. $\Pa(\star,\aangle{\ap_i',\star,\star},\sn)$
      is true. \Cref{line:pa-check-ap} of the code of Predicate $\Pa$ (\Cref{alg:pred}) implies that $\ap_i'$ is a valid AP at sequence number $\sn-1$ from $p_i$ for some payload $\aangle{A_\prover,\star}$. $A_\prover$ fulfills \Cref{def:valid:UA}, and is therefore a valid UA at sequence number $\sn-1$ from $p_i$.
      By induction, we obtain that, for each $\sn'$ s.t. $0 \leq \sn' \leq \sn$ there is a valid UA $A_i^{\sn'}$ at $\sn'$ from $p_i$. We denote this sequence as $\precedingUAs({A_i'}^\sn)=(A_i^0,A_i^1,..,A_i^\sn)$.
  \end{itemize}
  By \Cref{lemma:unicity:UAs}, $\precedingUAs({A_i'}^\sn)$ is unique.
  We say that the sequence of valid UAs $(A_i^{\sn'})_{0 \leq \sn' \leq \sn}$ is the \textit{sequence of preceding UAs} of $A_i^\sn$. Prefix ordering similarly follows from \Cref{lemma:unicity:UAs}.
\end{proof}

\begin{definition}[Valid transfer $\tx$ at a sequence number $\sn$ for a process $p_i$]\label{def:valid:transfer}
    We say that that a transfer $\tx=\aangle{\snd, v, \rcv, \sn_\snd}$ is valid at a sequence number $\sn>0$ for a process $p_i$ (correct of faulty) if there exist a non-membership proof $u$ and two valid UAs $A_i^{\sn}$ and $A_i^{\sn-1}$ issued by $p_i$ at sequence numbers $\sn$ and $\sn-1$ respectively, such that:
    \begin{equation}\label{eq:tx:validity}
        \accVerifyNonMem(A_i^{\sn-1},\tx,\nmp) \land \accAdd(A_i^{\sn-1},\tx) = A_i^{\sn}.
    \end{equation}
    When this holds, we note $\tx$'s sequence number and UA issuer\footnote{We note that the UA issuer of transfer $\tx$ is not always the sender of $\tx$.} as superscript and subscript, respectively, \ie $\tx=\tx_i^\sn$.
\end{definition}

When the context is clear, for simplicity, we might abbreviate \Cref{eq:tx:validity} using a set notation into $A_i^{\sn}=A_i^{\sn-1}\uplus\{\tx\}$, where $\uplus$ denotes the disjoint set union.

\begin{lemma}[Transfer associated with a valid UA]\label{lemma:trf:valid:UA}
    Consider a valid UA $A_i^\sn$ at sequence number $\sn>0$ issued by a process $p_i$ (correct or faulty)\footnote{The sequence number $\sn=0$ is excluded as the genesis AP $\ap_i^0$ does not have a corresponding transfer.}.
    The existence of $A_i^\sn$ implies that there exists a valid transfer $\tx_i^\sn$ issued by $p_i$ at sequence number $\sn$ (\Cref{def:valid:transfer}). Furthermore this transfer is unique for a given $p_i$ and $\sn$. We say that $A_i^\sn$ is associated with $\tx_i^\sn$.
\end{lemma}
\begin{proof}
    Consider a valid UA $A_i^\sn$ at sequence number $\sn>0$ issued by a process $p_i$ (correct or faulty).
    By \Cref{def:valid:UA}, there exists a valid AP $\ap_i^\sn$ at sequence number \sn for $p_i$ so that $\AccountUpdate.\apverify(\ap_i^\sn, \aangle{A_i^\sn,\star}, \sn, i)$.
    By AP-Knowledge-Soundness, and by definition of the $\Pa$ predicate for \AccountUpdate objects (\Cref{alg:pred}), the prover $p_i$ must know some $\data=\aangle{\star,\zkp_i,\olddata}$ satisfying $\Pa(\aangle{A_i^\sn,\star},\data,\sn)$.
    Due to \cref{line:pa-check-zkp} of the code of Predicate $\Pa$ (\Cref{alg:pred}), this implies that $\TransferValidity.\zkpverify(\zkp_i,\olddata\oplus\aangle{A_i^\sn,\star,\sn})$ is true, where $\olddata = \aangle{A_i^{\sn-1},\star}$.
    As a result, by ZKP-Knowledge-Soundness and by construction of the $\Pzk$ predicate of the ZKP object \TransferValidity (\Cref{alg:pred}), the prover $p_i$ must know some \secdata containing a transfer $\tx=\aangle{\snd,v,\rcv,\sn'}$ variable (fifth parameter at \cref{line:pzk-secdata-unwrap} of \Cref{alg:pred}) such that $\Pzk(\aangle{A_i^{\sn-1},\star} \oplus \aangle{A_i^\sn,\star,\sn},\secdata)$ is true (if $i=\snd$, then $\sn=\sn'$ due to the check at \cref{line:pzk-check-pvr-is-snd} of \Cref{alg:pred}, otherwise if $i=\rcv$, $\sn$ and $\sn'$ can be different). Furthermore, due to \cref{line:pa-check-ap}, and by \Cref{lemma:unicity:UAs}, $A_i^{\sn-1}$ is the valid UA at sequence number $\sn-1$ for process $p_i$.
    This observation, together with \cref{line:pzk-check-nonmem-trf} of \Cref{alg:pred} implies that $\tx$ is a valid transfer issued by $p_i$ at sequence number $\sn$. Because $A_i^\sn$ and $A_i^{\sn-1}$ are unique for $p_i$ at their respective sequence number, \cref{line:pzk-check-nonmem-trf} further yields that $\tx$ is the sole transfer that is valid for $p_i$ and $\sn$.
\end{proof}

\begin{corollary}[Sequence of preceding transfers of a valid UA]\label{lemma:seq:preceding:transfers}
    Consider a valid UA $A_i^\sn$ at sequence number $\sn\geq 0$ issued by a process $p_i$ (correct or faulty).
    $A_i^\sn$ can be associated with a unique \textit{sequence of preceding transfers}, denoted $\precedingTrans(A_i^\sn)=(\tx_i^1,..,\tx_i^\sn)$, so that each $\tx_i^k$ is a valid transfer issued by $p_i$ at sequence number $k$, with $1 \leq k \leq \sn$.
\end{corollary}
\begin{proof}
    The corollary follows from \Cref{lemma:preceding:UAs,lemma:trf:valid:UA}.
\end{proof}

\begin{lemma}[Prefix-ordering of of preceding transfers]
    The sequences of transfers produced for UAs issued by the same process $p_i$ are prefix-ordered, \ie if $A_i^{\sn_1}$ and $A_i^{\sn_2}$ are two valid UAs issued by $p_i$ at sequence numbers $\sn_1$ and $\sn_2$ respectively such that $\sn_1\leq\sn_2$, then $\precedingTrans(A_i^{\sn_1})$ is a prefix of $\precedingTrans(A_i^{\sn_2})$.
\end{lemma}
\begin{proof}
    This follows from the unicity of a valid transfer for a given process and sequence number, as stated in \Cref{lemma:trf:valid:UA}.
\end{proof}

For simplicity, in the following, we equate valid transfers with their corresponding \transfer invocation.
More precisely, if $\tx_i^\sn$ is a valid transfer at sequence number $\sn$ for a process $p_i$, we denote the \transfer invocation corresponding to a $\tx_i^\sn$ by $\transfer_\snd^{\sn'}(v,\rcv)$, where $\snd$ and $\sn'$ are respectively the id of the sender and the sequence number of the transfer contained in $\tx$.
In this case, we say that $\transfer_\snd^{\sn'}(v,\rcv)$ is the transfer invocation for prover $p_i$ at sequence number $\sn$.

\begin{definition}[Sending UAs, receiving UAs, and null UAs]
\label{def:sending:receiving:null:UAs}
We distinguish two sorts of valid UAs that are produced in our system: \textit{sending UAs}, and \textit{receiving UAs}.
\begin{itemize}
    \item A valid $A_i^\sn$ is a sending UA if 
    the issuer $p_i$ of $A_i^\sn$ is also the sender of the corresponding transfer $\transfer_i^\sn(\star,j)$ (where $i$ and $j$ can be different).
    If the issuer $p_i$ is correct, the sending UA has been produced at \cref{line:aat-process-add-trf} of \Cref{alg:aat}, during the $\prove()$ call at \cref{line:aat-trf-process}.

    \item A valid $A_i^\sn$ is a receiving UA if
    the issuer $p_i$ of $A_i^\sn$ is also the receiver of the corresponding transfer $\transfer_j^{\sn'}(\star,i)$ (where $i$ and $j$ can be different, and $\sn$ and $\sn'$ can be different).
    If the issuer $p_i$ is correct, the receiving UA is produced at \cref{line:aat-process-add-trf} of \Cref{alg:aat}, during the $\prove()$ call at \cref{line:aat-rcv-trf-process}.
\end{itemize}
Recall that a transfer invocation $\transfer_i^\sn(v,i)$ from a process $p_i$ to itself is allowed in our algorithm, which entails that the corresponding UA $A_i^\sn$ is both a sending UA and a receiving UA.
In this case, $A_i^\sn$ is said to be a \textit{null} UA, and $\transfer_i^\sn(v,i)$ is said to be a \textit{null} \transfer invocation. 
\end{definition}

\newcommand{\nonnullswitch}{non-null }
\renewcommand{\nonnullswitch}{}

\begin{lemma}[Matching \nonnullswitch{}receiving UA to a \nonnullswitch{}sending UA]\label{lemma:matching:receiving:UA}
Any valid \nonnullswitch{}receiving UA $A_i^\sn$ for some transfer $\tx_i=\aangle{j, \star, i, \sn_{\tx_i}}$ from a receiver $p_i$ (correct or faulty) can be matched to a valid \nonnullswitch{}sending UA $A_j^{\sn_{\tx_i}}$ for the same transfer $\tx_i$ from the sender $p_j\neq p_i$ (correct or faulty).
\end{lemma}

\begin{proof}
Consider a valid \nonnullswitch{}receiving UA $A_i^\sn$ (\Cref{def:sending:receiving:null:UAs}) at sequence number $\sn$ for process $p_i$ corresponding to some receiving transfer $\tx_i=\aangle{j, \star, i, \sn_{\tx_i}}$ whose sender is $p_j$ and receiver is $p_i$. Two cases arise depending on whether $A_i^\sn$ is a null or non-null UA (cf. \Cref{def:sending:receiving:null:UAs}).
\begin{itemize}
    
    \item If $A_i^\sn$ is a null UA, it is both a sending UA and a receiving UA, and corresponds to a transfer $\tx_i=\aangle{i, \star, i, \sn}$ from $p_i$ to itself, thus fulfilling the lemma. 
    
    \item If $A_i^\sn$ is a non-null receiving UA, this $A_i^\sn$ is associated to a valid AP $\ap_i^{\sn}$ (\Cref{def:valid:UA}). \Cref{line:pzk-valid-sending-AP-exists-start,line:pzk-valid-sending-AP-exists-end} of \Pzk (\Cref{alg:pred}) imply the existence of a valid UA $A_j^{\sn_{\tx_i}}$ for $p_j$ at sequence number $\sn_{\tx_i}$, with $j\neq i$ (since $A_i^\sn$ is non-null), and $\tx_i\in A_j^{\sn_{\tx_i}}$.
    
    Consider the sequence $\precedingUAs(A_j^{\sn_{\tx_i}})=(A_j^{k})_{0\leq k \leq \sn_{\tx_i}}$ of UAs preceding $A_j^{\sn_{\tx_i}}$ (\Cref{rem:seq:previous:UAs}), and the sequence $\precedingTrans(A_j^{\sn_{\tx_i}})=(\tx_j^k)_{0<k\leq \sn_{\tx_i}}$ of transfers preceding $A_j^{\sn_{\tx_i}}$ (\Cref{lemma:seq:preceding:transfers}). By construction of the two sequences, each transfer $\tx_j^k$ is valid for $p_j$ at sequence number $k$, which implies
    \begin{equation}
        \forall k\in [1..\sn_{\tx_i}]: A_j^{k}=A_j^{k-1}\uplus\{\tx_j^k\}.
    \end{equation}
    \Cref{line:pa-check-first-trf} of \Pa (\Cref{alg:pred}) implies that $A_j^0$ is empty. This observation and the above equation yield that $A_j^{\sn_{\tx_i}}$ contains exactly the transfers present in the sequence $\precedingTrans(A_j^{\sn_{\tx_i}})$. 
    
    Since $\tx_i\in A_j^{\sn_{\tx_i}}$, there exists some $k_0\in [1..\sn_{\tx_i}]$ such that $\tx_i=\tx_j^{k_0}$. As the sender of $\tx_i$ is $p_j$, \cref{line:pzk-check-pvr-is-snd} of \Pzk (\Cref{alg:pred}) implies that $k_0=\sn_{\tx_i}$, and that $\tx_i=\tx_j^{\sn_{\tx_i}}$ by the unicity of a valid transfer for a given process and sequence number, as stated in \Cref{lemma:trf:valid:UA}.\qedhere
\end{itemize}
\end{proof}

\mparagraph{Notations.}
For the remaining proofs, we introduce the following notations to manipulate the different notions we have discussed.
All notations are defined with respect to a particular global history $H$.
\begin{itemize}
    \item $\transferFunc(A_i^\sn) \defeq \transfer_\snd^{\sn'}(v,\rcv)$ is the transfer invocation corresponding to a valid UA $A_i^\sn$ when $\sn>0$ (see \Cref{lemma:trf:valid:UA}). 
    $A_i^\sn$ might be a sending UA (in which case $i=\snd$ and $\sn=\sn'$) or a receiving UA (in which case $i=\rcv$).
    
    \item $\sendingUA(\transfer_i^{\sn}(v,j)) \defeq A_i^\sn$ is the sending UA issued by $p_i$ corresponding to a valid transfer invocation $\transfer_i^{\sn}(v,j)$ performed by $p_i$ towards $p_j$. By definition of a valid transfer and \Cref{lemma:matching:receiving:UA}, this sending UA always exists, and by AP-Agreement, it is unique.
    \TA{Is AP-Agreement still relevant here?}
    As a shortcut, we further say that $\sn$ (the sequence number of the UA $A_i^\sn$) is the sequence number of the transfer invocation $\transfer_i^{\sn}(v,j)$.

    \newcommand{\transferTest}{\ensuremath{\mathsf{transfer}}}
    \newcommand{\transferTestTwo}{\ensuremath{\mathsf{transfer}}\xspace}
    \newcommand{\sendingUAtest}{\mathit{sendingAPtest}}
    
    \item $\sendingUA(A_i^{\sn}) \defeq A_j^{\sn'}$ is by extension the sending UA issued by $p_j$ corresponding to a valid receiving UA $A_i^{\sn}$ issued by $p_i$. $\sendingUA(A_i^{\sn})$ is a shortcut for $\sendingUA(\transferFunc(A_i^{\sn}))$.
    In case $A_i^{\sn}$ is a null UA (\Cref{def:sending:receiving:null:UAs}), we have $\sendingUA(A_i^{\sn})=A_i^{\sn}$.
    
    \item $\receivingUA(\transfer_i^{\sn}(v,j)) \defeq A_j^{\sn'}$ is, when it exists, the receiving UA issued by $p_j$ corresponding to a valid transfer invocation $\transfer_i^{\sn}(v,j)$ performed by $p_i$ towards $p_j$. When both the sender $p_i$ and the receiver $p_j$ are correct, this receiving UA always exists because of the network's reliability, and because the internal operation $\prove(..)$ of \Cref{alg:aat} (\crefrange{line:aat-prove-start}{line:aat-prove-end}) does not contain any blocking operation.
    
    When $\transfer_i^{\sn}(v,j)$ has no corresponding receiving UA (this can happen when either the sender $p_i$ or the receiver $p_j$ are Byzantine), by convention we define $\receivingUA(\transfer_i^{\sn}(v,j)) \defeq \bot$. 
    \item If a correct process $p_i$ performs a $\balancee_i()$ invocation while its $\sn_i$ variable has value \sn, we denote it by $\balancee_i^\sn()$.%
    \footnote{Recall that we do not consider $\balancee()$ invocations from Byzantine processes.}
\end{itemize}
For simplicity, we denote by $\op_i^\sn()$ any invocation $\balancee_i^\sn()$ or $\transfer_i^\sn(\star,\star)$.





\subsubsection{AT-Sequentiality of Algorithm~\ref{alg:aat}}
\label{sec:proof:at-seq}
Let us consider any execution of \Cref{alg:aat}, captured as a global history $H=(L_1,...,L_n)$. 
\ft{To prove AT-Sequentiality (\Cref{lemma:at-sequentiality}), we must construct a mock history $\mock{H}$ of $H$ that preserves the local histories of correct processes and replaces the local histories $L_j$ of Byzantine processes $p_j$ by mock local histories $\mock{L}_j$ so that $\mock{H}$ can be AT-sequenced. We construct $\mock{H}$ and the corresponding AT-sequence $\sequence_i$ incrementally by introducing several intermediary definitions and lemmas.}

In the following, $C$ is the set of correct processes, and $B$ the set of Byzantine processes.
$C\cup B=\{p_1,..,p_n\}$.
We construct the mock local histories $\{\mock{L}_j\}_{p_j\in B}$ of Byzantine processes as follows.

\newcommand{\UASet}{\Sigma}

\begin{definition}
We define the set $\UASet$ of UAs defined recursively by the following rules.
\begin{itemize}
    \item $\UASet$ contains all the UAs issued by correct processes at \cref{line:aat-process-add-trf} of \Cref{alg:aat}. This includes both sending and receiving UAs.
    \item If some receiving UA $A$ belongs to $\UASet$, then $\sendingUA(A)$ also belongs to $\UASet$,
    $$\big\{\sendingUA(A)\mid A \text{ is a valid receiving UA}\wedge A\in\UASet\big\} \subseteq \UASet.$$
    \item If some UA $A$ belongs to $\UASet$, then all UAs by the same issuer that precedes $A$ also belong to $A$,
    \begin{equation}\label{eq:precedingUAs:in:APSet}
    \bigcup_{A\in\UASet}\supportSet(\precedingUAs(A))\subseteq \UASet.    
    \end{equation}
\end{itemize}
$\UASet$ can be seen as the causal transitive closure of the UAs issued by correct processes: it contains all the UAs that can be traced back to some UA issued by a correct process, either through a sending/receiving relationship as captured by \Cref{lemma:matching:receiving:UA}, or through local precedence as captured by \Cref{rem:seq:previous:UAs}.
\end{definition}

\begin{lemma}[Validity of UAs in $\UASet$]\label{remark:validity:UAs:APSet}
    All UAs contained in $\UASet$ are valid.
\end{lemma}
\begin{proof}
    The remark holds because all UAs issued by correct processes are valid by construction, and results from the application of \Cref{lemma:matching:receiving:UA} and \Cref{rem:seq:previous:UAs} when defining $\sendingUA(\cdot)$ and $\precedingUAs(\cdot)$ respectively.
\end{proof}

\newcommand{\transactionSet}[1]{\Theta_{#1}}%

\begin{definition}[Mock history $\mock{H}$]\label{def:mock:H}
For a Byzantine process $p_j\in B$, we construct the set $\transactionSet{j}$ of transfer invocations whose sender is $p_j$, and for which both a sending and a receiving UA can be found in $\UASet$. More formally, we have
$$
    \transactionSet{j}\defeq\left\{ \transfer_j^{\sn}(\star,\star) \;\left|\;
    \begin{array}{c}
    \exists A_1,A_2\in\UASet: A_1 \text{ is a sending UA}\;\wedge\; 
    A_2 \text{ is a receiving UA}\;\wedge \\
    \;\transferFunc(A_1)=\transferFunc(A_2)=\transfer_j^{\sn}(\star,\star) \end{array}
    \right.\right\}.
$$
\newcommand{\SNSet}{\mathit{SN}}%
Let us note $\SNSet_j$ the set of the sequence numbers of the transfer invocations present in $\transactionSet{j}$. $\SNSet_j$ is a set of strictly positive integers (since $p_j$'s genesis UA $A_j^0$ created at \cref{line:aat-init} of \Cref{alg:aat-init-vars} does not have a corresponding transfer invocation), $\SNSet_j\subseteq \mathbb{N}^*$.
Note that, by AP-Agreement, each sequence number in $\SNSet_j$ corresponds to a single transfer invocation in $\transactionSet{j}$.
As a result, we can define the mock local history $\mock{L}_j$ of $p_j$ as the sequence of transfer invocations present in $\transactionSet{j}$ ordered by their sequence numbers,
$$
\mock{L}_j\defeq \big( \transfer_j^{\sn}(\star,\star)\in \transactionSet{j} \big)_{\sn\in \SNSet_j}.
$$
Combining the mock local histories $(\mock{L}_j)_{j\in B}$ of Byzantine processes with the local histories $(L_i)_{i\in C}$ of correct processes yields a mock history $\mock{H}$.
\end{definition}

\newcommand{\perceivedSet}{\mathcal{S}}

\begin{definition}[Set $\perceivedSet_i$ of operation perceived by a correct process $p_i$]
    Consider a correct process $p_i$, and the set $\perceivedSet_i$ containing all invocations made by $p_i$ (to the operations $\transfer$ and $\balancee$), and all transfer invocations by other processes (correct of faulty) present in $\mock{H}$.
    \begin{equation}\label{eq:perceivedSet_i}
    \perceivedSet_i=\supportSet(L_i)\cup\bigcup_{j\in C\setminus\{i\}}\transferSet(L_j)
    \cup\bigcup_{j\in B}\transferSet(\mock{L}_j),
    \end{equation}
    where $C$ and $B$ are the sets of correct and Byzantine processes as defined above, respectively, and $\transferSet(\cdot)$ is the notation introduced in \Cref{sec:amt-spec} to denote the set of transfer invocations contained in a sequence.
\end{definition}

    To produce a sequence $\sequence_i$ from the elements of $\perceivedSet_i$, we construct a partial order $\POATSeq$ on $\perceivedSet_i$. We do so incrementally, by introducing two binary relations, $\POATSeqOne$ and $\POATSeqTwo$. 
    \begin{itemize}
        \item $\POATSeqOne$ totally orders the invocation of $p_i$ and ensures that incoming transfers received by a process $p_j$ are ordered before the subsequent balance invocations (if $j=i$) and outgoing transfers issued by $p_j$. 
        \item $\POATSeqTwo$ focuses specifically on $p_i$ to guarantee that incoming transfers received by $p_i$ are totally ordered with respect to $p_i$'s local invocation to $\balancee(\cdot)$.
    \end{itemize}
    Distinguishing between $\POATSeqOne$ and $\POATSeqTwo$ as intermediary steps towards $\POATSeq$ will in turn make it easier to prove that the $\POATSeq$ is acyclic (a necessary condition to show that it is a partial order).

\begin{definition}[Binary relations $\POATSeqOne$, $\POATSeqTwo$, and $\POATSeq$ on $\perceivedSet_i$]
    We define $\POATSeqOne$ on $\perceivedSet_i$ as follows.
    \begin{itemize}
        \item {\bf Inclusion of $\bm{p_i}$'s local order.} $\POATSeqOne$ respects the local process order $\rightarrow_i$ of operations invoked by $p_i$, \ie $\rightarrow_i\subseteq\POATSeqOne$.

        \newcommand{\transferToJ}{\transTo{j}}
        \newcommand{\transferFromJ}{\transFrom{j}}
    
        \item {\bf Ordering of incoming transfers.} 
        \begin{itemize}
            \item For any process $p_j$, $\transferToJ=\transfer_{\star}^{\star}(\star,j)\in \perceivedSet_i$ a \nonnullswitch{}transfer invocation whose receiver is $p_j$, and $\transferFromJ=\transfer_{j}^{\star}(\star,\star)\in \perceivedSet_i$ a \nonnullswitch{}transfer invocation by $p_j$ such that $\transferToJ\neq\transferFromJ$, if $\receivingUA(\transferToJ)\in \precedingUAs(\sendingUA(\transferFromJ))$, then $\transferToJ\POATSeqOne \transferFromJ$.
            
            \item For \nonnullswitch{}transfer invocation $\transTo{i}=\transfer_{\star}^{\star}(\star,i)\in \perceivedSet_i$ whose receiver is $p_i$, and any invocation $b_i=\balancee_i()/v$ performed by $p_i$, if the execution of \cref{line:aat-process-add-trf} that generated $\receivingUA(\transTo{i})$ occurred before that of \cref{line:aat-blc} corresponding to $b_i$ , then $\transTo{i}\POATSeqOne b_i$.     
        \end{itemize}
    \end{itemize}
    We define $\POATSeqTwo$ on $\perceivedSet_i$ as follows.
    \begin{itemize}
           \item {\bf Ordering of $\bm{p_i}$'s balance.} For any invocation $b_i=\balancee_i()/v$ performed by $p_i$, and any \nonnullswitch{}transfer invocation $\transTo{i}=\transfer_{\star}^{\star}(\star,i)\in \perceivedSet_i$ whose receiver is $p_i$, if the execution of \cref{line:aat-blc} corresponding to $b_i$ occurred before that of \cref{line:aat-process-add-trf} that generated $\receivingUA(\transTo{i})$, then $b_i\POATSeqTwo \transTo{i}$.
    \end{itemize}
    $\POATSeq$ is defined as the transitive closure of the union of $\POATSeqOne$ and $\POATSeqTwo$, \ie
    $
    \POATSeq = \big(\POATSeqOne\cup\POATSeqTwo\big)^{+}
    $.
\end{definition}

\begin{lemma}\label{subremark:POATSeqOne:acyclic}
$\POATSeqOne$ is acyclic.
\end{lemma}

\begin{proof}
\newcommand{\timestamp}{\tau}
The distributed system is implicitly embedded in some physical time, which we assume is Newtonian and linear.\footnote{
    Note that the proof also holds in the context of special relativity, even though there may not be an absolute global clock ordering the events of the system.
    This is because any observer in the system still sees the events in a linear (total) order that may, however, differ from one observer to the next.
    For this proof, it suffices to take any observer (\eg $p_i$) as a frame of reference to construct the timestamps.
    But because of the causality principle, two causally-linked events are necessarily seen in the same order, no matter the chosen frame of reference.
}
We associate each element $x$ of $\perceivedSet_i$ with a timestamp $\timestamp(x)$ from this physical time as follows:
\begin{itemize}
   \item if $x=\balancee_i()/v$ is a balance invocation by $p_i$, $\timestamp(x)$ is the time point at which \cref{line:aat-blc} of \Cref{alg:aat} is executed;
   
   \item if $x=\transfer_{\star}^{\star}(\star,\star)$ is a transfer invocation, $\timestamp(x)$ is the time point at which $\sendingUA(x)$ was created.
\end{itemize}
Consider $x,y\in \perceivedSet_i$ such that $x\POATSeqOne y$.
At least one of the following cases holds.
\begin{itemize}
   \item {\it Case 1 (Inclusion of $p_i$'s local order):} $x$ and $y$ are both invoked by $p_i$ with $x\rightarrow_i y$. Because $p_i$ is sequential, by definition of $\rightarrow_i$ we have $\timestamp(x)<\timestamp(y)$.
   
   \item {\it Case 2 (Ordering of incoming transfers):} 
   \begin{itemize}
       \item $x=\transfer_{\star}^{\star}(\star,j)\in \perceivedSet_i$ is a \nonnullswitch{}transfer invocation whose receiver is some process $p_j$, and $y=\transfer_{j}^{\star}(\star,\star)\in \perceivedSet_i$ a \nonnullswitch{}transfer invocation by $p_j$, such that $\receivingUA(x)\in \precedingUAs(\sendingUA(y))$. This last inclusion means that $\sendingUA(y)$ can be linked recursively to $\receivingUA(x)$ using proofs issues by $p_j$, the predicates $\Pa$ and $\Pzk$, and AP-Knowledge-Soundness and ZKP-Knowledge-Soundness. AP-Knowledge-Soundness and ZKP-Knowledge-Soundness both imply that $p_j$ must know the secret data required to issue a proof {\it before} the proof is generated. As a result $\receivingUA(x)\in \precedingUAs(\sendingUA(y))$ implies that $\timestamp(x)<\timestamp(y)$.
       \item $x=\transfer_{\star}^{\star}(\star,i)\in \perceivedSet_i$ is a \nonnullswitch{}transfer invocation whose receiver is $p_i$, $y=\balancee_i()/v$ is a $\balancee$ invocation performed by $p_i$, such that the execution of \cref{line:aat-process-add-trf} that generated $\receivingUA(x)$ occurred before that of \cref{line:aat-blc} corresponding to $y$. The same reasoning as above yields that $\timestamp(x)<\timestamp(y)$. 
   \end{itemize}
\end{itemize}

We conclude that $\POATSeqOne$ respects the timestamps assigned by $\timestamp(\cdot)$. Assuming physical time is acyclic, this implies that $\POATSeqOne$ is also acyclic.
\end{proof}
    
\begin{lemma}\label{subremark:POATSeqTwo:acyclic}
$\POATSeqTwo$ is acyclic.
\end{lemma}
\begin{proof}
$\POATSeqTwo$ is a simple bipartite binary relation between \balancee invocations on one side, and \transfer invocations on the other, hence it cannot contain cycles, and is therefore acyclic.
\end{proof}

\begin{lemma}\label{subremark:one:cup:two:acyclic}
$\POATSeqOne\cup\POATSeqTwo$ is acyclic.
\end{lemma}

\begin{proof}
\newcommand{\Cmin}{\mathcal{C}_{\mathsf{min}}}
The proof is by contradiction. Assume $\POATSeqOne\cup\POATSeqTwo$ contains at least one cycle. Consider the shortest such cycle $\Cmin$. Because $\POATSeqOne$ and $\POATSeqTwo$ are both acyclic (\Cref{subremark:POATSeqOne:acyclic,subremark:POATSeqTwo:acyclic}), $\Cmin$ must involve both $\POATSeqOne$ and $\POATSeqTwo$.

Consider one of the $\POATSeqTwo$ ordering in $\Cmin$. By definition of $\POATSeqTwo$, this ordering is of the form $b_i\POATSeqTwo \transTo{i}$, where $b_i$ is a balance invocation by $p_i$ and any $\transTo{i}$ is a \nonnullswitch{}transfer invocation whose receiver is $p_i$. Consider the successor $y$ of $\transTo{i}$ in $\Cmin$. By construction of $\POATSeqTwo$, $\transTo{i}\not \POATSeqTwo y$, which implies $\transTo{i} \POATSeqOne y$.

By definition of $\POATSeqOne$, $\transTo{i} \POATSeqOne y$ leads to two cases.
\begin{itemize}
    \item {\em Case 1:} $\transTo{i}\rightarrow_i y$. In this case, $y\in L_i$ by definition of $\rightarrow_i$, $p_i$'s local process order.
    \item {\em Case 2:}  $\transTo{i}\not\rightarrow_i y$. In this case, $\transTo{i} \POATSeqOne y$ resulted from the rule \emph{Ordering of incoming transfers}. $y$ is therefore either a $\balancee$ or $\transfer$ operation invoked by $p_i$, \ie we also have $y \in L_i$.
\end{itemize} 

We therefore have $b_i\POATSeqTwo \transTo{i} \POATSeqOne y$, with $b_i, y \in L_i$.
The definitions of $\POATSeqTwo$ and $\POATSeqOne$ imply that the invocation of $b_i$ by $p_i$ precedes in time the generation of $\receivingUA(\transTo{i})$ also by $p_i$, which itself precedes the invocation of $y$ still by $p_i$ (either directly in the case of $\POATSeqTwo$, and if $y$ is a $\balancee$ invocation, or using the same argument on cryptographic proofs as in \Cref{subremark:POATSeqOne:acyclic} if $y$ is a $\transfer$).
By definition of the process order $\rightarrow_i$, this leads to $b_i\rightarrow_i y$, and hence $b_i\POATSeqOne y$.
We can, therefore, remove $\transTo{i}$ from $\Cmin$, and replace $b_i\POATSeqTwo \transTo{i} \POATSeqOne y$ by $b_i\POATSeqOne y$, thus creating a cycle  $\Cmin'$ in $\POATSeqOne\cup\POATSeqTwo$ with one less element than $\Cmin$.
This is a contradiction as $\Cmin$ was assumed to be the shortest.
\end{proof}


\begin{corollary}
$\POATSeq$ is a partial order on $\perceivedSet_i$.
\end{corollary}

\begin{proof}
$\POATSeq$ is transitive by construction. By \Cref{subremark:one:cup:two:acyclic}, it is the transitive closure of an acyclic relation and is, therefore, also acyclic. 
\end{proof}

\begin{lemma}\label{subremark:poset:finite:prede}
Each element of the poset ($\perceivedSet_i$,$\POATSeq$) only has a finite number of predecessors.
\end{lemma}

\begin{proof}
This follows from the sequential nature of processes (and in particular that all processes have a starting point in their execution), the link between $\POATSeqOne$ and $\POATSeqTwo$ and physical time (\Cref{subremark:POATSeqOne:acyclic,subremark:POATSeqTwo:acyclic}), and the (implicit) assumption that processes only take a finite number of local steps per unit of time.
\end{proof}

\begin{definition}[Sequence $\sequence_i$]\label{def:sequence:i}
By \Cref{subremark:poset:finite:prede} the poset ($\perceivedSet_i$,$\POATSeq$) can be sorted topologically into a sequence. We define $\sequence_i$ as one of the topological sorts of  ($\perceivedSet_i$,$\POATSeq$). We note $\rightarrow_{\sequence_i}$ the total order induced by $\sequence_i$ on its elements.
\end{definition}

In the following, we will use the function $\total(\cdot,\cdot)$ introduced in \Cref{sec:amt-spec} to compute the account balance of a process resulting from the transfer it sends and receives. Given a process $p_j$ and a set of operations invocations $O$, $\total(j,O)$ is defined as
$$
\definitionTotal{j}
$$

\begin{lemma}\label{lemma:balance}
$\Forall b_i = \balancee_i()/v \in \sequence_i: v = \total(i,\{o \in \sequence_i \mid o \rightarrow_{\sequence_i} b_i\}).$
\end{lemma}

\begin{proof}
Consider $b_i= \balancee_i()/v\in \sequence_i$ a balance operation performed by a correct process $p_i$ that returns a value $v$. 
%
Define the sets $I_i$ and $R_i$ as follows:
\begin{itemize}
    \item $I_i$ is the set of \nonnullswitch{}$\transfer$s \emph{invoked} locally by $p_i$ before $p_i$ executed $b_i$.
    By construction of \Cref{alg:aat}, $p_i$ generated a sending UA at \cref{line:aat-process-add-trf} for each of these transfers.
    
    \item $R_i$ is the set of \nonnullswitch{}$\transfer$s \emph{received} by $p_i$ whose receiving UA was generated by $p_i$ (at \cref{line:aat-process-add-trf}) before invoking $b_i$.
\end{itemize} 

Because $p_i$ is correct, it executes \Cref{alg:aat}, which ensures by construction (in particular due to \cref{line:aat-process-update-bal-start,line:aat-process-update-vars,line:aat-process-update-bal,line:aat-blc}) that
\begin{equation}\label{eq:v:total}
    v=\total(i,I_i \cup R_i),
\end{equation}
taking into account that the effects of null transfers on $i$ cancel out in the definition of $\total$.
%
\newcommand{\AllTransfersToI}{T_{\vartriangleright i}}%
\newcommand{\AllTransfersFromI}{T_{i\vartriangleright}}%

Consider now $\AllTransfersFromI$ the subset of \nonnullswitch{}transfers present in $\perceivedSet_i$ that $p_i$ issued, and $\AllTransfersToI$ the subset of \nonnullswitch{}transfers present in $\perceivedSet_i$ that $p_i$ received, \ie
\begin{align}
    \AllTransfersFromI\defeq&\{\transfer_i(\star,k) \in \perceivedSet_i \mid i\neq k\}, \\
    \AllTransfersToI\defeq&\{\transfer_k(\star,i) \in \perceivedSet_i\mid i\neq k\}.
\end{align}
By definition of $\total(\cdot,\cdot)$, we have similarly
\begin{equation}\label{eq:total:i:sequence:i}
    \total\big(i,\{o \in \sequence_i \mid o \rightarrow_{\sequence_i} b_i\}\big)=\total\big(i,\{o \in \sequence_i \mid o \rightarrow_{\sequence_i} b_i\}\cap (\AllTransfersFromI\cup\AllTransfersToI)\big).
\end{equation}

Because $p_i$'s local process order is included in $\POATSeq$, it is also included in $\rightarrow_{\sequence_i}$.
As a result, the \nonnullswitch{}transfers invoked by $p_i$ before invoking $b_i$ are exactly those \nonnullswitch{}transfers that precedes $b_i$ in $\rightarrow_{\sequence_i}$,
$$\{o \in \sequence_i \mid o \rightarrow_{\sequence_i} b_i\}\cap \AllTransfersFromI=I_i.$$

Furthermore, due to the rule \emph{Ordering of incoming transfers} of $\POATSeqOne$ and the definition of $\POATSeqTwo$, $b_i$ is totally ordered w.r.t. to the \nonnullswitch{}transfers received by $p_i$ in $\POATSeq$ and therefore in $\rightarrow_{\sequence_i}$.
As a result, we have
$$\{o \in \sequence_i \mid o \rightarrow_{\sequence_i} b_i\}\cap \AllTransfersToI=R_i.$$
These last two equalities together with \Cref{eq:total:i:sequence:i,eq:v:total} yield the lemma.\footnote{
    Let us recall that, given 3 sets $A$, $B$, and $C$, we always have $(A \cap B) \cup (A \cap C) = A \cap (B \cup C)$.
}
\end{proof}

\begin{lemma}\label{lemma:transfer}
$\Forall j\in[0..n], \Forall \transFrom{j} = \transfer_j(v,\star) \in \sequence_i: v \leq \total\big(j,\{o \in \sequence_i \mid o \rightarrow_{\sequence_i} \transFrom{j}\}\big).$
\end{lemma}
\begin{proof}
Consider a process $p_j$ (correct or faulty) and $\transFrom{j} = \transfer_j(v,\star)\in \sequence_i$ a $\transfer$ invocation from $j$ that is present $\sequence_i$ of amount $v$ to some process.


\newcommand{\sentBeforeInSeqI}{I_{\transFrom{j}}^{\sequence_i}}
\newcommand{\receivedBeforeInSeqI}{{R_{\transFrom{j}}^{\sequence_i}}}

Let us define the sets $\sentBeforeInSeqI$ and $\receivedBeforeInSeqI$ as follows.
\begin{itemize}
    \item $\sentBeforeInSeqI$ is the set of transfers whose sender is $p_j$ that appear before $\transFrom{j}$ in sequence $\sequence_i$ and includes $\transFrom{j}$,
    \begin{equation*}
        \sentBeforeInSeqI = \big\{\transfer^\star_j(\star,\star)\in\sequence_i\mid \transfer^\star_j(\star,\star)\rightarrow_{\sequence_i}\transFrom{j} \vee \transfer^\star_j(\star,\star)=\transFrom{j}\big\}.
    \end{equation*}
    
    \item $\receivedBeforeInSeqI$ is the set of transfers whose receiver is $p_j$ that appear before $\transFrom{j}$ in sequence $\sequence_i$.
    \begin{equation}\label{eq:receivedBeforeInSeqI}
        \receivedBeforeInSeqI = \big\{\transfer^\star_\star(j,\star)\in\sequence_i\mid \transfer^\star_\star(j,\star)\rightarrow_{\sequence_i}\transFrom{j}\big\}.
    \end{equation}
\end{itemize}

\begin{subremark}\label{subrem:total:B4InSeqI}
$\total(j,\sentBeforeInSeqI\cup\receivedBeforeInSeqI)= \total(j,\{o \in \sequence_i \mid o \rightarrow_{\sequence_i} \transFrom{j}\}) - v.$
\end{subremark}

\begin{proof}
This directly follows from the definition of $\total(\cdot,\cdot)$, $\sentBeforeInSeqI$ and $\receivedBeforeInSeqI$.
\end{proof}

\noindent
\newcommand{\SNsentBeforeInSeqI}{\mathit{SN}_{\transFrom{j}}^{\sequence_i}}
Consider $\SNsentBeforeInSeqI$ the set of the sequence numbers of the transfers present in $\sentBeforeInSeqI$, \ie
\begin{equation}\label{eq:SNsentBeforeInSeqI}
    \SNsentBeforeInSeqI=\{\sn\in\mathbb{N}^* \mid \transfer^\sn_j(\star,\star)\in \sentBeforeInSeqI\}.
\end{equation}
Because $\transFrom{j}\in\sentBeforeInSeqI$, $\SNsentBeforeInSeqI$ contains the sequence number of $\transFrom{j}$.
However, because $\POATSeq$ does not impose strong constraints on the order of the transfer invocations issued by a process other than $p_i$, $\sentBeforeInSeqI$ might also contain transfer invocations that have a higher sequence number than $\transFrom{j}$.
Furthermore, note that, by construction, $\SNsentBeforeInSeqI$ only contains sequence numbers of transfers whose sender is $p_j$.

\newcommand{\maxSn}{\sn^{\mathsf{max}}_{j\vartriangleright}}
\newcommand{\maxAP}{A^{\mathsf{max}}_{j\vartriangleright}}
Define $\maxSn$ as the highest sequence number contained in $\SNsentBeforeInSeqI$,
\begin{equation}\label{eq:maxSn}
    \maxSn=\mathsf{max}\big(\SNsentBeforeInSeqI\big),
\end{equation}
and $\maxAP$ as the UA issued by $j$ with the sequence number $\maxSn$, \ie $\maxAP=A_j^{\maxSn}$.
(By AP-Agreement, this UA is unique.)

\newcommand{\UABeforeMax}{\Sigma^{\mathsf{max}}_{j\vartriangleright}}
Consider $\UABeforeMax$ the set of UAs issued by $p_j$ that precede $\maxAP$ according to the predicates $\Pa$ and $\Pzk$ (\Cref{rem:seq:previous:UAs}),
\begin{equation}\label{eq:APBeforeMax}
    \UABeforeMax = \supportSet\big(\precedingUAs(\maxAP)\big).
\end{equation}
\newcommand{\sentBeforeInAP}{I_{\transFrom{j}}^{\mathsf{max}}}
\newcommand{\receivedBeforeInAP}{{R_{\transFrom{j}}^{\mathsf{max}}}}
Based on $\UABeforeMax$, let us define the sets $\sentBeforeInAP$ and $\receivedBeforeInAP$ as follows.
\begin{itemize}
    \item $\sentBeforeInAP$ is the set of transfers whose sender is $p_j$ and whose sending UA appear in $\UABeforeMax$,
    \begin{equation}\label{eq:sentBeforeInAP}
        \sentBeforeInAP = \big\{\transfer^\star_j(\star,\star) \mid \sendingUA(\transfer^\star_j(\star,\star))\in \UABeforeMax \big\}.
    \end{equation}
    \item $\receivedBeforeInAP$ is the set of transfers whose receiver is $p_j$ and whose receiving UA appear in $\UABeforeMax$, excluding the transfer invocation associated with $\maxAP$ in case it is a null UA (\ie both a sending and receiving UA),
    \begin{equation} \label{eq:receivedBeforeInAP}
        \receivedBeforeInAP = \left\{\transfer^\star_\star(j,\star) \left| 
        \begin{array}{l}\receivingUA(\transfer^\star_\star(j,\star))\in \UABeforeMax \wedge\\
        \transfer^\star_\star(j,\star)\neq\transferFunc(\maxAP)
        \end{array}%
        \right.\right\}.
    \end{equation}
\end{itemize}

\newcommand{\vmax}{v_{\mathsf{max}}}
Let $\vmax$ denote the amount of the transfer invocation $\transferFunc(\maxAP)$, \ie $\transferFunc(\maxAP)=\transfer_j^{\maxSn}(\vmax,\star)$.
(The  notation $\transferFunc(\maxAP)$  was defined at the end of \Cref{sec:proof-prelim})

\begin{subremark}\label{subrem:totalB4AP:eq:minus:vmax}
$\total\left(j,\sentBeforeInAP \cup \receivedBeforeInAP\right) = \total\left(j,\left(\sentBeforeInAP \setminus \left\{\transferFunc(\maxAP)\right\}\right) \cup \receivedBeforeInAP\right) - \vmax.$
\end{subremark}

\begin{proof}
Similarly to \Cref{subrem:total:B4InSeqI}, the remark follows from the fact that $\transferFunc(\maxAP)\in\sentBeforeInAP$ and $\transferFunc(\maxAP)\not\in\receivedBeforeInAP$, the definition of $\vmax$, and that of $\total(\cdot,\cdot)$.
\end{proof}

\begin{subremark}\label{subrem:totalB4AP:geq:vmax}
$\total\left(j,\left(\sentBeforeInAP\setminus\left\{\transferFunc(\maxAP)\right\}\right)\cup\receivedBeforeInAP\right)\geq\vmax.$
\end{subremark}

\begin{proof}
This follows from a recursion on the sequence of (valid) UAs $\precedingUAs(\maxAP)$ that precede $\maxAP$ (\Cref{rem:seq:previous:UAs}), as well as the checks and updates contained in the $\Pzk$ predicate (\Cref{alg:pred}). This includes the check at \cref{line:pzk-sending-transfer-same-sn} that a sending process must have enough funds, and the updates performed on the sender's balance at \cref{line:pzk-check-snd-is-rcv} (for null UAs), \cref{line:pzk-check-pvr-is-snd} (for non-null sending UAs), and \cref{line:pkz-end-check-balance} (for non-null receiving UAs).
\end{proof}

\begin{subremark}\label{subremark:totalAP:geq:0}
    $\total\left(j,\sentBeforeInAP\cup\receivedBeforeInAP\right)\geq 0.$
\end{subremark}
\begin{proof}
    This is a direct consequence of \Cref{subrem:totalB4AP:eq:minus:vmax,subrem:totalB4AP:geq:vmax}.
\end{proof}

\begin{subremark}\label{eq:sentBeforeInSeqI:subseteq:sentBeforeInAP}
    $\sentBeforeInSeqI\subseteq\sentBeforeInAP.$
\end{subremark}
\begin{proof}
Consider $\transFrom{j}'=\transfer^{\sn'}_j(\star,\star)\in \sentBeforeInSeqI$, a transfer invocation whose sender is $p_j$ at sequence number $\sn'$ that belongs to $\sentBeforeInSeqI$. By definition of $\SNsentBeforeInSeqI$ and $\maxSn$ (\Cref{eq:SNsentBeforeInSeqI,eq:maxSn}), we have
\begin{equation*}
    \sn'\leq\maxSn.
\end{equation*}

By AP-Agreement, this inequality yields $\sendingUA(\transFrom{j}')\in\supportSet(\precedingUAs(\maxAP))$, and by definition of $\UABeforeMax$ and $\sentBeforeInAP$ (\Cref{eq:APBeforeMax,eq:sentBeforeInAP}), that $\transFrom{j}'\in\sentBeforeInAP$, proving the remark.
\end{proof}

\begin{subremark} \label{eq:receivedBeforeInAP:subseteq:receivedBeforeInSeqI}
$\receivedBeforeInAP\subseteq\receivedBeforeInSeqI.$
\end{subremark}

\begin{proof}
Consider $\transTo{j}'=\transfer^\star_\star(j,\star)\in \receivedBeforeInAP$, a transfer invocation whose receiver is $p_j$ that belongs to $\receivedBeforeInAP$. By definition of $\UABeforeMax$ and $\receivedBeforeInAP$ (\Cref{eq:APBeforeMax,eq:receivedBeforeInAP}), 
\begin{equation}\label{eq:condition:for:POATSeqOne}
    \receivingUA(\transTo{j}') \in \supportSet\big(\precedingUAs(\maxAP)\big) \;\wedge\; \transTo{j}' \neq \transferFunc(\maxAP).  
\end{equation}

Furthermore, because $\transferFunc(\maxAP) \in \perceivedSet_i$, we have $\transTo{j}' \in \perceivedSet_i$ (either because the sender of $\transTo{j}$ is correct, or by construction of $\UASet$ if it is faulty, see \Cref{eq:perceivedSet_i,eq:precedingUAs:in:APSet}).
Remember that $\maxAP$ is associated with sequence number $\maxSn\in\SNsentBeforeInSeqI$ (\Cref{eq:maxSn}). Because $\SNsentBeforeInSeqI$ only contains the sequence numbers of transfers whose sender is $p_j$, $\transferFunc(\maxAP)$ must be some transfer $\transfer^{\maxSn}_j(\star,\star)$ sent from $p_j$.
As a result, $\transTo{j}' \in \perceivedSet_i$ together with \Cref{eq:condition:for:POATSeqOne} mean that the condition for the rule \emph{Ordering of incoming transfers} in the definition of $\POATSeqOne$ is met, yielding $                \transTo{j}'\POATSeqOne\transferFunc(\maxAP)$, and therefore $                \transTo{j}'\POATSeq\transferFunc(\maxAP)$, and 
\begin{equation} \label{eq:trans':seq_i:maxAP}
    \transTo{j}' \rightarrow_{\sequence_i} \transferFunc(\maxAP).
\end{equation}

Because $\transferFunc(\maxAP)\in \receivedBeforeInSeqI$, by definition of $\receivedBeforeInSeqI$ (\Cref{eq:receivedBeforeInSeqI}), $\transferFunc(\maxAP)\rightarrow_{\sequence_i}\transFrom{j}$, which with \Cref{eq:trans':seq_i:maxAP} leads to $\transTo{j}'\rightarrow_{\sequence_i}\transFrom{j}$, and therefore $\transTo{j}'\in \receivedBeforeInSeqI$.
\end{proof}

\begin{subremark}\label{subremark:totalSeqI:geq:totalAP}
$\total(j,\sentBeforeInSeqI\cup\receivedBeforeInSeqI)\geq \total\left(j,\sentBeforeInAP\cup\receivedBeforeInAP\right)$
\end{subremark}

\begin{proof}
This follows from \Cref{eq:sentBeforeInSeqI:subseteq:sentBeforeInAP,eq:receivedBeforeInAP:subseteq:receivedBeforeInSeqI} and the definition of $\total(\cdot,\cdot)$.
\end{proof}

The lemma follows from \Cref{subrem:total:B4InSeqI,subremark:totalAP:geq:0,subremark:totalSeqI:geq:totalAP}.
\end{proof}

\atsequentiality*
\begin{proof}
The lemma is proved by using $\mock{H}$ provided by \Cref{def:mock:H}, the sequence $\sequence_i$ constructed in \Cref{def:sequence:i} for each correct process $p_i$, and considering \Cref{lemma:transfer,lemma:balance} that show that $\sequence_i$ is an AT-Sequence.
\end{proof}

\subsubsection{AT-Termination of Algorithm~\ref{alg:aat}}
\attermination*

\begin{proof}
The \balancee operation of \Cref{alg:aat} terminates trivially.
Furthermore, the only blocking instruction of the \transfer operation of \Cref{alg:aat} is the \approve operation at \cref{line:aat-process-create-ap}, in the \prove internal operation (called at \cref{line:aat-trf-process}).
This \approve operation terminates by AP-Termination, so the \transfer operation also terminates.
\end{proof}

\subsection{Proof of quasi-anonymity} \label{sec:proof-qaat}

\qaatrcvanonymity*

\begin{proof}
To determine if the adversary \Adv can obtain any information on the receiver $p_j$ of any $\transfer_i(\star,j)$ invocation given a global history $H$, we analyze the content of all message types produced by \Cref{alg:aat}.
Our algorithm features 2 message types with the following content:
\begin{enumerate}
    
    \item Accumulator update sent by a (sender or receiver) process $p_i$ publicly to the network in the \approve operation at \cref{line:aat-process-create-ap}: value $v=\aangle{A_i',\balcom_i'}$ and data $\data=\aangle{\ap_i, \zkp, \sn_i+1, \aangle{A_i, A'_i, \balcom_i, \balcom'_i}}$;
    
    \item Message from a sender $p_\snd$ to a receiver $p_\rcv$: $\aangle{\tx, \ap_\snd, A_\snd, \mp_\snd, \balcom_\snd}$.
\end{enumerate}

Note that sender-receiver messages (item 2) are exchanged 
using the \rasend operation, which ensures receiver anonymity and confidentiality (see Section~\ref{sec:model}).
Therefore, only accumulator updates (item 1) need to be considered.
C-Hiding guarantees that the commitments $\balcom_i$ and $\balcom'_i$ do not leak any data on the corresponding committed values.
Similarly, UA-Indistinguishability guarantees that the accumulators $A_i$ and $A'_i$ do not leak any data on their accumulated sets, and ZKP-Zero-Knowledge ensures that $\zkp$ does not leak any data on its secret input $\secdata$.
Finally, the sequence number $\sn_i+1$ does not reveal whether $p_i$ is a sender or receiver, only the number of transfers (debits or credits) in $A_i$.
Thus, \Adv does not learn any information from accumulator updates.

We conclude that without any information on the content of transfers already included or newly added to some history $H$, \Adv cannot deduce the recipient of a transfer more accurately than random, \ie with a probability of less than $\min((1/n)+\epsilon(\lambda),1)$, where $\epsilon$ is the statistical negligible function and $\lambda$ is the security parameter of the system.
\end{proof}

\qaatconfidentiality*

\begin{proof}
The proof follows the same reasoning as that of \Cref{lem:qaat-rcv-anonymity}.
\end{proof}

\section{Implementations of the Schemes of Section~\ref{sec:techniques} / Appendix~\ref{sec:techniques-long}} \label{sec:impl}

\subsection{Preliminaries}

\subsubsection{Groups of unknown order}
Our choice of UA~\cite{BBF19} implementation requires the use of groups of unknown order (as defined in~\cite{DK02}) in which the Strong RSA assumption, the Low Order assumption and the Adaptive Root assumption holds. 
Concrete examples of groups of unknown order are RSA groups, ideal class groups and hyperelliptic Jacobians. However, only class groups and hyperelliptic Jacobians are efficient in the trustless setup model. 
In particular,~\cite{DGS22} provides algorithms to trustlessly setup genus-$3$ hyperelliptic Jacobians of unknown order with $3392$-bits field elements for $128$ bits of security, which is close to the $3072$-bits elements of RSA groups generated using a trusted setup for the same bit-security. They also provide a new compression algorithm for class groups, reaching $128$-bits of security with $5088$-bits elements.

\subsection{Accumulator}
\label{sec:acc-impl}

For our accumulator (specified in \Cref{sec:accumulator}), we opt for the trustless extension of universal RSA accumulators in groups of unknown order from~\cite{BBF19}. While RSA accumulators typically require the accumulated elements to be prime numbers, one can use a hash function with prime domain to accumulate arbitrary elements. 
The scheme features constant size public parameters, accumulator digest and proofs of (non-)membership. Moreover, (non-)membership proofs can be verified in constant time. The scheme can thus be described as succinct.
Note that when adding a new element $e$ to an accumulator $A$, the membership proof of $e$ is the value of $A$ before $e$ was added. Therefore, as a process only needs the membership proof of a transfer when updating the accumulator and creating the ZKP of its account update, membership proof generation is done in constant time in~\Cref{alg:aat}. 

\subsubsection{Hashing to primes}
An RSA accumulator can only contain prime numbers.
To overcome this limitation, one must use a hash-to-prime function $H_{\Primes}:\{0,1\}^* \rightarrow \Primes(\lambda)$, where $\Primes(\lambda)$ is the set of primes smaller than $2^\lambda$. $H_{\Primes}$ must be deterministic and collision-resistant to prevent an element from being mapped to multiple primes or multiple elements from being hashed into the same prime.
A common method of hashing to prime numbers is to combine a pseudo-random function (\eg hash function) with a probabilistic primality test. The resulting function consists of successively hashing the concatenation of the data to hash with an incremented counter until finding the smallest nonce that yields a prime number.

\subsubsection{Experimentations}
All our tests are executed sequentially on a core i9-11950H processor. Our implementation is coded in C++ using the GNU MP Bignum library.
We instantiated $H_{\Primes}$ using SHA3 and a probabilistic Miller-Rabin primality test with $80$ rounds, resulting in an error probability of $4^{-80}$.
The output of SHA3 is truncated to $264$ bits, which is sufficient to obtain $128$ bits of security as the prime counting function estimates that there are at least $2^{256}$ primes in $[0,2^{264}]$. We measure an average time of $9$ ms to hash a single element to a $264$ bits prime.
We also implement a $2048$-bit RSA accumulator. \Cref{impl:acc} shows the performance of proof generation and verification. We observe that while (non-)membership proof generation from "scratch" scales linearly with the size of the accumulated set, membership proof generation during element addition is computed in constant time as element accumulation consists of a single exponentiation of the RSA digest by the newly added element. Similarly, we observe that (non-)membership proof verification is constant time, as it also consists of a single group exponentiation. 



\def\arrvline{\hfill\kern\arraycolsep\vline\kern-\arraycolsep\hfilneg}
\begin{table*}
    \centering
    \begin{threeparttable}
        \begin{tabular}{|c|c|c|c|c|c|}
        \hline
            \multirow{2}{*}{$|S|$} & \multirow{2}{*}{Addition\tnote{1}} & \multicolumn{2}{c|}{Membership proof} & \multicolumn{2}{c|}{Non-membership proof} \\ \hhline{|~|~|-|-|-|-|}
                                 &                           & Generation\tnote{2} & Verification & Generation & Verification \\ \hhline{|=#=|=|=|=|=|}
            \multicolumn{1}{|c||}{$100$} & $0.337$ ms & $31.462$ ms & $0.351$ ms & $31.835$ ms & $0.687$ ms \\ \hhline{|-|-|-|-|-|-|}
            \multicolumn{1}{|c||}{$500$} & $0.368$ ms & $155.366$ ms & $0.345$ ms & $151.720$ ms & $0.687$ ms \\ \hhline{|-|-|-|-|-|-|}
            \multicolumn{1}{|c||}{$1000$} & $0.342$ ms & $299.615$ ms & $0.339$ ms & $298.825$ ms & $0.675$ ms \\ \hhline{|-|-|-|-|-|-|}
            \multicolumn{1}{|c||}{$1500$} & $0.337$ ms & $447.879$ ms & $0.338$ ms & $444.978$ ms & $0.663$ ms \\ \hhline{|-|-|-|-|-|-|}
            \multicolumn{1}{|c||}{$2000$} & $0.360$ ms & $607.562$ ms & $0.345$ ms & $605.435$ ms & $0.668$ ms \\ \hhline{|-|-|-|-|-|-|}
        \end{tabular}
        \begin{tablenotes}
            \item[1] {\footnotesize This also outputs the membership proof of the new element.}
            \item[2] {\footnotesize Unused, as we rely on the proof generated during addition.}
        \end{tablenotes}
    \end{threeparttable}
    \caption{Performance of accumulator proof generation and verification for various set sizes $|S|$.}
    \label{impl:acc}
\end{table*}

\subsection{Commitment scheme} \label{sec:com-impl}
In our system, each process creates a hiding commitment of a single value (the balance of its account). This commitment can be based on any constant-size commitment scheme with trustless setup. For example, one could use Pedersen commitments in groups of unknown order (or elliptic curves~\cite{FG17}), which are unconditionally hiding and computationally binding. Conversely, other schemes such as ElGamal commitments~\cite{E85} could be used to obtain commitments that are unconditionally binding but computationally hiding. Note that it is impossible for a commitment scheme to be both unconditionally binding and unconditionally hiding.

\subsection{Transparent zk-SNARK with time-optimal prover} \label{sec:zkp-impl}
For our zk-SNARK scheme, we choose Spartan~\cite{S20}, which is both transparent (trustless setup) and prover time-optimal.

\section{A consensus-free quorum-based Agreement Proof implementation} \label{sec:ap-implem}
In this section, we provide an implementation of the agreement proof scheme of \Cref{sec:agreement-proof} that does not rely on consensus, uses succinct cryptography, and has a storage of $O(\lambda+n)$ bits per process and an overall communication of $O(\lambda n)$ bits.
Our AP implementation leverages threshold digital signatures (see \Cref{sec:thresh-sigs}) to create constant-size (\ie $O(\lambda)$-bit) agreement proofs $\ap$ (\ie quorums of signatures) and to verify them without the set of all public keys.






\subsection{Threshold signatures} \label{sec:thresh-sigs}
\ar{Threshold signatures are a family of aggregated digital signatures that are produced in a system of $n$ processes, by having at least $k$ out of $n$ processes (the threshold) produce individual signatures for the same message, which we call \textit{intermediary signatures}.
Once these intermediary signatures have been produced, they are aggregated (typically by a coordinating process) into a fixed-size aggregated signature, called the threshold signature $\sigma$.
We call a threshold signature with a threshold of $k$ signers out of $n$ processes a \textit{$k$-threshold signature}.
Unlike multi-signature schemes, which keep track of the identities of their signers, classical threshold signature schemes use a protocol called distributed key generation (DKG) to distribute the secret key of a signature scheme across processes. Processes then interact to reconstruct a single signature that can be verified against a single system-wide public key.
The computation of the individual shares of the secret key and the derivation of the system-wide public key are typically performed during the setup phase of the system.
Naturally, a $k$-threshold signature for a message $m$ is valid if and only if it aggregates $k$ valid distinct intermediary signatures for $m$.

Various threshold signature schemes have been built on top of asymmetric schemes such as RSA, Schnorr, BLS or ECDSA~\cite{BLS04,SS19,KG20,RRJSS22,DXKR23,WMYC23}.
Furthermore, the verification algorithm of most threshold versions of a classical signature scheme remains unchanged and can thus be efficiently verified independently of the number of co-signers using the system-wide public key.
For example, the asynchronous DKG protocol of~\cite{DXKR23} can be used with DLOG-based threshold cryptosystems such as BLS~\cite{BLS04} or Schnorr~\cite{RRJSS22,BHKMR24} to provide $k$-threshold signatures with high-threshold values ($k\in [t,n-t-1]$).

\subsubsection{Succinctness of threshold signatures} \label{sec:tsig-succinct}

Efficient high-threshold signature implementations (\eg \cite{BLS04,RRJSS22,BHKMR24}) guarantee that, for a fixed $\lambda$ and any number of signers $k$, the size and verification time of a threshold signature is equivalent to that of a single intermediary signature.
Unlike classical signature schemes and multi-signature schemes,  which require knowing the public keys of all the signers of a quorum of signatures, threshold signature schemes only require knowing one global public key for the system for verifying a quorum of signatures.
In other words, given the maximum number of signers $n$ in the system, instead of having quorum signatures of size $O(\lambda n)$ bits with classical signature schemes or of size $O(\lambda+n)$ bits with multi-signature schemes (as the set of signers of a quorum can be encoded in a bit vector), threshold signatures can produce quorum signatures of size $O(\lambda)$ bits (once a quorum has been generated, the individual identities of the signers are not needed for its verification).
Therefore, these implementations are of constant size and constant verification time, and thus succinct.}

\subsubsection{Verifying intermediary signatures without storing all individual public keys} \label{sec:tsig-constant-pki}
As mentioned previously, only a single global public key is needed to verify a threshold signature, instead of all the signers' public keys.
However, the coordinator must verify all intermediary signatures' validity before generating the threshold signature.
\ta{To do that without having to store all the system's public keys, it is possible to use a \emph{vector commitment} (\emph{VC})~\cite{CF11}, a generalization of the classical commitment scheme presented in \Cref{sec:commitment}.
Like its name suggests, a VC scheme commits a vector of values, upon which it is then possible to have a ``partial'' opening of a single value of the vector, without having to open the vector in its entirety. \ar{One of the most widely known examples of vector commitment is Merkle trees~\cite{M87}.}

In our context, this scheme would allow us to ``compress'' the entire vector of public keys of the system into a VC of size $O(\lambda)$.
This VC would be created during the system setup.
In addition to the global VC, each correct process $p_i$ also stores its individual public key, and the partial opening for the $i$\textsuperscript{th} value of the VC, corresponding to $p_i$'s public key.
This way, $p_i$ can prove to another process $p_j$ that it indeed belongs to the system, by sending to $p_j$ its individual public key and partial opening.
With this data, $p_j$ can then verify that the $i$\textsuperscript{th} value of the global VC indeed opens to $p_i$'s public key.

Since each correct process $p_i$ stores the global VC ($O(\lambda)$ bits), its individual public key ($O(\lambda)$ bits), and the partial opening of the $i$\textsuperscript{th} value of the global VC ($O(\lambda)$ bits), this technique does not impact the overall storage cost of our AP algorithm (\Cref{alg:ap}) of $O(\lambda+n)$ bits stored per correct process.

More generally, this technique effectively allows us to have a public key infrastructure (PKI) under a fixed storage cost of $O(\lambda)$ bits per correct process (provided that the set of public keys does not change over time).
For simplicity, all this machinery is hidden in the validation of intermediary signatures of \Cref{alg:ap}, \cref{line:ap-rcv-quoruminit-ver-sig,line:ap-rcv-quorumsig-ver-sig}.
}

\ta{
\subsubsection{Efficiently aggregating intermediary signatures} \label{sec:tsgi-efficient-agg}
For constructing a $k$-threshold signature, a coordinator process has to receive and aggregate $k$ intermediary signatures by $k$ distinct processes.
In particular, it has to make sure that it does not aggregate multiple intermediary signatures from the same process.

For the sake of simplicity, we address these issues by a naive approach in \Cref{alg:ap}: we require the coordinator $p_i$ to save all the intermediary signatures it receives (and implicitly their signer's identity) in the $\sigs_i$ set (\cref{line:ap-rcv-quorumsig-save-sig}).
Once the threshold signature has been produced, $\sigs_i$ is emptied, and all the temporary data (intermediary signatures and sender identities) are deleted.
The $\sigs_i$ set contains at most $n$ intermediary signatures ($O(\lambda)$ bits) and signer identities ($O(\log n)$ bits), which incurs a total of $O(n(\lambda+\log n))$ bits.
However, as this data is stored only temporarily during the execution of the $\approve()$ operation, we consider that it pertains to the \emph{space complexity} of $\approve()$, and not the overall \emph{storage cost} of our system.

Remark that this suboptimal space complexity can be easily reduced to $O(\lambda+n)$ bits using the following method.
Instead of aggregating the intermediary signatures all at once at the end of the execution of $\approve()$ (\cref{line:ap-prove-aggreg}), they can be gradually aggregated as they are received (\cref{line:ap-rcv-quorumsig}) on an aggregated signature of $O(\lambda)$ bits, which becomes a threshold signature once it aggregates at least $k$ intermediary signatures. 
This way, each intermediary signature can be forgotten right after it is received and aggregated.
Furthermore, the coordinator can use a bit vector of $n$ bits to efficiently keep track of the processes that already participated in the aggregated signature (the $i$\textsuperscript{th} bit indicates if the intermediary signature of $p_i$ was already received or not).
}

\subsection{Algorithm}

In this section, we present our AP implementation, which eschews consensus by using threshold signatures for quorums.
It achieves lightness by using succinct cryptographic primitives and ensuring a storage cost of $O(\lambda+n)$ bits per correct process, and a communication of $O(\lambda n)$ bits sent overall.\footnote{
    For these analyses, like in our asset transfer algorithm of \Cref{alg:aat}, we impose that sequence numbers are constant-size.
}

\Cref{alg:ap} describes the code of this implementation for a correct process $p_i$ and a AP predicate \Pa.
It uses 3 local variables: $\sigs_i$ (a set of intermediary signatures used for generating threshold signatures), $v_i$ (the value of the current \approve execution), and $\seqnums_i$ (the array containing the sequence number (initialized with $0$) of each process $p_j$, $j \in [n]$).

The \apverify operation simply verifies that the provided AP $\ap_j$ is a valid $(\lfloor\frac{n+t}{2}\rfloor{+}1)$-threshold signature for the provided value $v$, sequence number $\sn_j$ and process $p_j$ (\cref{line:ap-init}).

The \approve operation first computes its next sequence number $\sn_i$ of $p_i$ (\cref{line:ap-prove-inc-sn}).
Then, it verifies that the provided $v$, \data, and $\sn_i$ satisfy the AP predicate (\cref{line:ap-prove-ver-pred}).
If it does, $p_i$ then saves the provided value $v$ in the local variable $v_i$ (\cref{line:ap-prove-save-v}).
Next, $p_i$ generates the first signature of its payload, called the initialization signature (\cref{line:ap-prove-sign}), and broadcasts it in a \quoruminitm message which requests other processes of the system to sign its payload (\cref{line:ap-prove-bcast}).
Process $p_i$ then waits for a quorum of intermediary signatures from other processes (\cref{line:ap-prove-wait}), and when it is received, $p_i$ aggregates all intermediary signatures into a $(\lfloor\frac{n+t}{2}\rfloor{+}1)$-threshold signature (\cref{line:ap-prove-aggreg}), flushes the temporary data (\cref{line:ap-prove-flush}) and returns the AP and next sequence number (\cref{line:ap-prove-ret}).

When $p_i$ receives a \quoruminitm message from a process $p_j$, it first checks that the provided initialization signature is valid (\cref{line:ap-rcv-quoruminit-ver-sig}) and that the provided $v$, \data, and $\sn_j$ satisfy the AP predicate (\cref{line:ap-rcv-quoruminit-ver-pred}).
Then, $p_i$ waits for the provided $\sn_j$ to be the next one (in FIFO order) to be processed (\cref{line:ap-rcv-quoruminit-wait}).
After that, $p_i$ produces its intermediary signature (\cref{line:ap-rcv-quoruminit-sign}) and sends it to $p_j$ in a \quorumsigm message (\cref{line:ap-rcv-quoruminit-send}).
Finally, $p_i$ updates the sequence number of $p_j$ in $\seqnums_i$ with its new $\sn_j$ (\cref{line:ap-rcv-quoruminit-update-com}).

When $p_i$ receives a \quorumsigm message, it checks if the intermediary signature is valid (\cref{line:ap-rcv-quorumsig-ver-sig}), if it has not already received a previous intermediary signature from the same signer (\cref{line:ap-rcv-quorumsig-ver-signer}), and if it has not already produced an AP for the payload (\cref{line:ap-rcv-quorumsig-ver-ap}).
If these conditions pass, $p_i$ saves the provided intermediary signature (\cref{line:ap-rcv-quorumsig-save-sig}).

\begin{algorithm}[!tb]
\Init{%
$\sigs_i \gets \varnothing$;
$v_i \gets \bot$;
$\seqnums_i[1..n] \gets [0..0]$.
} \label{line:ap-init}
\AlgoSkip

\Operation{$\apverify(\ap_j,v,\sn_j,j)$}{%
    \return $\left\{\begin{array}{ll}
    \ttrue & \text{if $\ap_j$ is a valid $(\lfloor\frac{n+t}{2}\rfloor{+}1)$-threshold signature for $\aangle{v,\sn_j,j}$}, \\
    \ffalse & \text{otherwise.}
    \end{array}\right.$ \label{line:ap-verify}
}
\AlgoSkip

\Operation{$\approve(v,\data)$}{
    $\sn_i \gets \seqnums_i[i]+1$; \label{line:ap-prove-inc-sn} \Comment*{increment own sequence number}
    \lIf{$\Pa(v,\data,\sn_i)=\ffalse$}{\return \abort;} \label{line:ap-prove-ver-pred}
    $v_i \gets v$; \label{line:ap-prove-save-v} \Comment*{save $v$ for condition at \cref{line:ap-rcv-quorumsig}}
    $\sig_i \gets \text{initialization signature of $\aangle{v,\data,\sn_i,i}$ by $p_i$}$; \label{line:ap-prove-sign} \\
    \broadcast $\quoruminitm(v,\data,\sn_i,i,\sig_i)$; \label{line:ap-prove-bcast} \\
    \wait $(\text{$\sigs_i$ has strictly more than $\frac{n+t}{2}$ intermediary signatures for $\aangle{v,\sn_i,i}$})$; \label{line:ap-prove-wait} \\
    $\ap_i$ $\gets$ aggregation of $\sigs_i$ into a ($\lfloor\frac{n+t}{2}\rfloor{+}1$)-threshold signature for $\aangle{v,\sn_i,i}$; \label{line:ap-prove-aggreg} \\
    $\sigs_i \gets \varnothing$;
    $v_i \gets \bot$; \label{line:ap-prove-flush} \Comment*{flush temporary data}
    \return $\langle \ap_i,\sn_i \rangle$. \label{line:ap-prove-ret}
}
\AlgoSkip

\When{$\textup{\quoruminitm}(v,\data,\sn_j,j,\sig_j)$ {\bf is} \receivedd}{
    \lIf{$\sig_j$ is not a valid initialization signature for $\aangle{v,\data,\sn_j,j}$ by $p_j$}{\return;} \label{line:ap-rcv-quoruminit-ver-sig}
    \lIf{$\Pa(v,\data,\sn_j)=\ffalse$}{\return;} \label{line:ap-rcv-quoruminit-ver-pred}
    \wait $\seqnums_i[j] = \sn_j-1$; \label{line:ap-rcv-quoruminit-wait} \\
    $\sig_i \gets \text{intermediary signature for $\langle v,\sn_j,j \rangle$ by $p_i$}$; \label{line:ap-rcv-quoruminit-sign} \\
    \send $\quorumsigm(i,\sig_i)$ to $p_j$; \label{line:ap-rcv-quoruminit-send} \\
    $\seqnums_i[j] \gets \sn_j$. \label{line:ap-rcv-quoruminit-update-com} \Comment*{update $p_j$'s sequence number}
}
\AlgoSkip

\When{$\textup{\quorumsigm}(j,\sig_j)$ {\bf is} \receivedd}{ \label{line:ap-rcv-quorumsig}
    \lIf{$\sig_j$ is not a valid intermediary signature for $\aangle{v_i,\sn_i,i}$ by $p_j$}{\return;} \label{line:ap-rcv-quorumsig-ver-sig}
    \lIf{some intermediary signature by $p_j$ already in $\sigs_i$}{\return;} \label{line:ap-rcv-quorumsig-ver-signer}
    \lIf{some AP $\ap_i$ for $v_i$ already produced by $p_i$}{\return;} \label{line:ap-rcv-quorumsig-ver-ap}
    $\sigs_i \gets \sigs_i \cup \{\sig_j\}$. \label{line:ap-rcv-quorumsig-save-sig}
    
}
\caption{Light and consensus-free quorum-based Agreement Proof algorithm for an AP predicate \Pa, and assuming $n>3t$ (code for $p_i$).}
\label{alg:ap}
\end{algorithm}


\subsection{Proof of Algorithm~\ref{alg:ap}}
\todo{Revise \Cref{alg:ap} proof with new definition/interface of AP object (where the sequence number if passed to the predicate.)}
In the following correctness proofs of the AP scheme, we consider an AP object $\mathit{Obj}_A$ set up with an AP predicate \Pa.

\begin{lemma}[AP-Validity]
If $\ap_i$ is a valid AP for a value $v$ at sequence number $\sn_i$ from a correct process $p_i$, then $p_i$ has executed $\mathit{Obj}_A$$.\approve(v,\star)/\ap_i$ as its $\sn_i$\textsuperscript{th} invocation of $\mathit{Obj}_A$$.\approve(..)$.
\end{lemma}

\begin{proof}
Let us assume that $\ap_i$ is a valid AP for value $v$ at sequence number $\sn_i$ from a correct process $p_i$.
By definition of the \apverify operation at \cref{line:ap-verify}, $\ap_i$ is a valid $(\lfloor\frac{n{+}t}{2}\rfloor+1)$-threshold signature for $v$ at $\sn_i$ from $p_i$.
This threshold signature must have been aggregated from at least $\lfloor\frac{n+t}{2}\rfloor+1$ intermediary signatures.
Given the system assumption $n>3t$, we have $\lfloor\frac{n+t}{2}\rfloor+1>\lfloor\frac{4t}{2}\rfloor+1=2t+1 \geq t+1$.
Therefore, at least one correct process must have produced an intermediary signature for $\aangle{v,\sn_i,i}$ at \cref{line:ap-rcv-quoruminit-sign}.
However, to execute this line, this correct process must have verified the initialization signature of $p_i$ at \cref{line:ap-rcv-quoruminit-ver-sig}.
By the unforgeability of signatures, the only way to produce this initialization signature is for $p_i$ to execute \cref{line:ap-prove-sign}, during an $\mathit{Obj}_A$$.\approve(v,\star)/\ap_i$ execution.
\end{proof}

\begin{lemma}[AP-Agreement]
There are no two different valid APs $\ap_i$ and $\ap_i'$ for two different values $v$ and $v'$ at the same sequence number $\sn_i$ and from the same prover $p_i$.
More formally, $\mathit{Obj}_A.\apverify(\ap_i,v,\sn_i,i)=\mathit{Obj}_A.\apverify(\ap_i',v',\sn_i,i)=\ttrue$ implies $v=v'$.
\end{lemma}

\begin{proof}
Let us assume, on the contrary, that there exists two different valid APs $\ap_i$ and $\ap_i'$ for two different values $v$ and $v'$ at the same sequence number $\sn_i$ and from the same process $p_i$ (correct or faulty).
By definition of the \apverify operation at \cref{line:ap-verify}, $\ap_i$ and $\ap_i'$ are valid $(\lfloor\frac{n+t}{2}\rfloor{+}1)$-threshold signatures for $v$ and $v'$ (resp.) at $\sn_i$ from $p_i$.
These threshold signatures must have been aggregated from the intermediary signatures of two sets of processes, $A$ and $B$, which have respectively signed $\aangle{v,\sn_i,i}$ and $\aangle{v',\sn_i,i}$.
We have $|A| \geq \lfloor\frac{n+t}{2}\rfloor+1 \leq |B|$, or equivalently, $|A| > \frac{n+t}{2} < |B|$.\footnote{Recall that $\Forall i \in \mathbb{Z}, r \in \mathbb{R}: (i \geq \lfloor r \rfloor + 1) \iff (i>r)$.\label{footnote:floor-equiv}}
We thus have $|A \cap B| = |A|+|B|-|A \cup B| > 2\frac{n+t}{2}-|A \cup B| \geq 2\frac{n+t}{2}-n=t$.
Therefore, at least one correct process $p_j$ must belong both to $A$ and $B$, and must have signed both $\aangle{v,\sn_i,i}$ and $\aangle{v',\sn_i,i}$ for $v \neq v'$.
But by the fact that correct processes produce intermediary signatures for some $\sn_i$ at \cref{line:ap-rcv-quoruminit-sign} and right after updating the $i$-th sequence number in their $\seqnums$ array at \cref{line:ap-rcv-quoruminit-update-com}, it follows that correct processes produce at most one intermediary signature for a given $\sn_i$ and sender identity $i$, which contradicts the fact that $p_i$ must belong both to $A$ and $B$.
\end{proof}






\begin{lemma}[AP-Knowledge-Soundness]
If $\ap_i$ is a valid AP for value $v$ at sequence number $\sn_i>0$ from a prover $p_i$ (correct or faulty), then $p_i$ knows some \data such that $\Pa(\star,\data,\sn_i)/\ttrue$.
\end{lemma}

\begin{proof}
Let us assume that $\ap_i$ is a valid AP for value $v$ at sequence number $\sn_i$ from a process $p_i$ (correct or faulty).
By definition of the \apverify operation at \cref{line:ap-verify}, $\ap_i$ is a valid $(\lfloor\frac{n+t}{2}\rfloor{+}1)$-threshold signature for $v$ at $\sn_i$ from $p_i$.
This threshold signature must have been aggregated from at least $\lfloor\frac{n+t}{2}\rfloor+1$ intermediary signatures.
Given the system assumption $n>3t$, we have $\lfloor\frac{n+t}{2}\rfloor+1>\lfloor\frac{4t}{2}\rfloor+1=2t+1 \geq t+1$.
Therefore, at least one correct process must have produced an intermediary signature for $\aangle{v,\sn_i,i}$ at \cref{line:ap-rcv-quoruminit-sign}.
However, to execute this line, this correct process must have verified the satisfaction of $\Pa$ by the received $\data$ at \cref{line:ap-rcv-quoruminit-ver-pred}, which entails that the prover $p_i$ must have known the $\data$ such that $\Pa(\star,\data,\sn_i)/\ttrue$.
\end{proof}

\TA{Maybe add illustration}

\begin{lemma}[AP-Termination]
Given a correct process $p_i$ that executes $\mathit{Obj}_A.\approve(v,\data)/r$ with value $v$ as its $\sn_i$\textsuperscript{th} invocation of $\mathit{Obj}_A$$.\approve(..)$, and a \data, if $\Pa(\star,\data,\sn_i)=\ttrue$ then $r=\ap_i$, where $\ap_i$ is a valid AP for $v$ at $\sn_i$ from $p_i$. If $\Pa(\star,\data,\sn_i)=\ffalse$, then $r=\abort$.
\end{lemma}

\begin{proof}
Let us assume that a correct process $p_i$ executes $\mathit{Obj}_A.\approve(v,\data)/r$ with value $v$ and a \data as its $\sn_i$\textsuperscript{th} invocation of $\mathit{Obj}_A$$.\approve(..)$.
If $\Pa(\ft{\star},\data,\ft{\sn_i})=\ffalse$, then $p_i$ passes the condition at \cref{line:ap-prove-ver-pred} and return \abort.
If $\Pa(\star,\data,\sn_i)=\ttrue$, $p_i$ continues the execution, produces an initialization signature at \cref{line:ap-prove-sign}, broadcasts a \quoruminitm message at \cref{line:ap-prove-bcast} and wait for a quorum of signatures from the system at \cref{line:ap-prove-wait}.

Upon receiving this \quoruminitm message, each correct process $p_j$ passes the conditions at \cref{line:ap-rcv-quoruminit-ver-sig} (as $p_i$ has correctly generated its initialization signature $\sig_i$) and at \cref{line:ap-rcv-quoruminit-ver-pred} (as $p_i$ has verified the satisfaction of \Pa by \data at \cref{line:ap-prove-ver-pred}), and then wait that the received $\sn_i$ is the next in FIFO order to be processed at \cref{line:ap-rcv-quoruminit-wait}.
As $p_i$ uses its sequence numbers in FIFO order, then $p_i$ has necessarily sent a \quoruminitm message for each $\sn_i' \in [1..\sn_i]$, and $p_j$ will eventually receive all these \quoruminitm messages.
By induction, $p_j$ will pass the \wait statement at \cref{line:ap-rcv-quoruminit-wait} for each $\sn_i' \in [1..\sn_i]$, because if $\sn_i'=1$, the condition \cref{line:ap-rcv-quoruminit-wait} is satisfied as $\seqnums_j[i]$ is initialized at \cref{line:ap-init} to $0$, and if $\sn_i'>1$, then $p_j$ will eventually replace the $\sn_i'-2$ by $\sn_i'-1$ at \cref{line:ap-rcv-quoruminit-update-com}.
\TA{Does it need more justification?}
Therefore, all correct processes, which are at least $n-t$, will eventually pass the \wait statement at \cref{line:ap-rcv-quoruminit-wait}.
Given the system assumption $n>3t$, we have $n-t = \frac{2n-2t}{2} > \frac{n+3t-2t}{2} = \frac{n+t}{2}$.
Therefore, strictly more than $\frac{n+t}{2}$ (or, equivalently, at least $\lfloor \frac{n+t}{2} \rfloor+1$)
correct processes will produce an intermediary signature at \cref{line:ap-rcv-quoruminit-sign} and send it back to $p_i$ at \cref{line:ap-rcv-quoruminit-send}.

Finally, $p_i$ will receive this quorum of intermediary signatures at \cref{line:ap-rcv-quorumsig}, which will unlock the \wait instruction at \cref{line:ap-prove-wait}, and $p_i$ will aggregate all intermediary signatures into a valid AP $\ap_i$, and will return $\aangle{\ap_i,\sn_i}$ at \cref{line:ap-prove-ret}.
\end{proof}



\subsection{Consensus-freedom of Algorithm \ref{alg:ap}} \label{sec:ap-implem-cons-free}
The consensus-freedom of \Cref{alg:ap} follows trivially from the fact that the only communication primitives that it uses are classic \send/\receive operations and a best-effort \broadcast operation, and because it always terminates in the presence of failures and asynchrony.

\subsection{Succinctness of Algorithm \ref{alg:ap}} \label{sec:ap-implem-succinctness}
The succinctness of \Cref{alg:ap} comes from the succinctness of threshold signatures: as shown in \Cref{sec:thresh-sigs}, the best implementations of threshold signatures guarantee that the size and verification time of threshold signatures (and therefore also the agreement proofs $\ap_i$ produced by \Cref{alg:ap}) are equivalent to that of a single intermediary signature.
Therefore, \Cref{alg:ap} is succinct.

\subsection{Storage cost of Algorithm~\ref{alg:ap}} \label{sec:ap-implem-storage}
The storage of \Cref{alg:ap} is of $O(\lambda+n)$ bits stored by each correct process.
It comes from its local variables: the $\sigs_i$ set and $v_i$ values are emptied at the end of each \approve execution (\cref{line:ap-prove-flush})\footnote{
    \ta{See \Cref{sec:tsgi-efficient-agg} for a discussion on why we consider the storage of these temporary variables to be part of the space complexity of the $\approve()$ operation (and not the system's storage cost), and how this space complexity could be easily reduced to $O(\lambda+n)$ bits.}
}, and $\seqnums_i$ is an array of size $n$ containing constant-size sequence numbers (by definition).
\ta{Lastly, the threshold signature scheme requires the storage of $O(\lambda)$ bits (see \Cref{sec:tsig-constant-pki}).}
Therefore, each correct process only stores $O(\lambda+n)$ bits in \Cref{alg:ap}.

\subsection{Communication cost of Algorithm~\ref{alg:ap}} \label{sec:ap-implem-comm}
The communication cost of \Cref{alg:ap} is $O(n\lambda)$ bits sent by correct processes overall during an execution of \approve.
It follows from the message exchanges entailed by an $\approve(v,\data)$ invocation by a correct prover $p_i$.
Recall that, for simplicity, we assume that the size of the $v$ and \data parameters of $\approve(..)$ are of constant size.

In a first step, $p_i$ broadcasts a $\quoruminitm(v,\data,\sn_i,i,\sig_i)$ message at \cref{line:ap-prove-bcast}, where $v$, \data, and $\sn_i$ are constant-size, $i$ has $O(\log n)$ bits, and $\sig_i$ has $O(\lambda)$ bits.
As this message is broadcast to all $n$ processes, the overall communication cost of this first step is $O(n(\lambda+\log n))$ bits.

In a second step, every correct process of the system $p_j$ receives the $\quoruminitm$ message, passes the conditions at \cref{line:ap-rcv-quoruminit-ver-sig,line:ap-rcv-quoruminit-ver-pred}, and sends a $\quorumsigm(j,\sig_j)$ message to $p_i$ at \cref{line:ap-rcv-quoruminit-send}, where $j$ has $O(\log n)$ bits, and $\sig_j$ has $O(\lambda)$ bits.
This amounts to $O(n(\lambda+\log n))$ bits sent overall by correct processes during this second step.

Therefore, the total number of bits sent by correct processes during the execution of \Cref{alg:ap} is $O(n(\lambda+\log n))$.
However, as we assume that $\lambda = \Omega(\log n)$ (see \Cref{sec:model}, \P~\textit{Security parameter $\lambda$}), the previous asymptote simplifies to $O(n\lambda)$ bits sent overall.

\section{Trustless system setup} \label{sec:aat-setup}
Like most distributed or cryptographic systems, our system requires a setup phase to compute the initial public and private parameters enabling processes to participate in the system.
We assume the initial knowledge of a genesis data structure to specify the process identifiers and the initial amounts of their respective accounts.
Our system is set up in a trustless manner, \ie no trusted party is involved in the distributed computation of the system parameters nor holds a trapdoor or secret that could otherwise be used to compromise the system's safety.
We refer to the distributed setup operation of the system as $\mathsf{system\_setup}()$.
Note that the generation of the genesis data is not part of the $\mathsf{system\_setup}()$.
This trustless setup is possible because the cryptographic schemes implementing our schemes also provide a trustless setup (see \Cref{sec:impl}).

We provide the following implementation draft to demonstrate the feasibility of $\mathsf{system\_setup}()$.
\begin{enumerate}
    \item Assume public knowledge of $\mathsf{genesis\_data} \leftarrow \bigcup_{i=1}^{n}\aangle{pk_i, A_i, \balcom_i}$, where $A_i$ is the empty accumulator of process $p_i$ and $\balcom_i$ is a commitment to the initial balance of $p_i$ ($\init_i$).
    
    \item Assume private knowledge of the opening of $\balcom_i$, secret parameters of $A_i$ and secret key of signature key pair $\aangle{sk_i, pk_i}$ for each process $p_i$.
    
    \item Execute the setup operation of each cryptographic scheme.
    
    
    \item Generate a threshold signature $\ap_i$ (compatible with \Cref{alg:ap}) that will serve as the initial agreement proof of $p_i$ for each $v_i=\aangle{A_i,\balcom_i}$.
    
    \item Each process $p_i$ can now safely delete $\mathsf{genesis\_data}$ and output its initial parameters $\aangle{A_i, \balcom_i, \ap_i, \balopen_i}$.
\end{enumerate}

\section{Transfer batching} \label{sec:batching}
In this section, we briefly propose a method to commit to batches of transfers (instead of single transfers) and to verify those batches succinctly, \ie faster than if the batched transfers were verified individually.
Such a method is advantageous both in terms of anonymity and efficiency, which is explained in more detail in \Cref{sec:enhancements}.

\mparagraph{Incrementally Verifiable Computation (IVC).}
Informally, an IVC can be represented as a function $F$ that takes as input a previous execution of $F$ and some additional input.
For example, $F$ could be the function implementing the operation $\prove()$ of \Cref{alg:aat}.
Then each execution of $F$ would take as input an accumulator, an agreement proof and a balance commitment of a process $p_i$ and update the accumulator and balance commitment of $p_i$ accordingly.
Note that by ``chaining'' several iterations of $F$, we obtain the processing of a batch of transfers. 

\mparagraph{Folding scheme for IVC proofs.}
Folding schemes such as Nova~\cite{KST22} or Sangria~\cite{GWC19} allow the creation of efficient SNARKS that a given number of iterations of an IVC were correctly executed by generating a proof for each step, then combining them successively on-the-fly in constant time (factor of $2$ for Nova).
Once the prover wishes to demonstrate the correct processing of its IVC iterations, the IVC proof is compressed into a single, small and succinct SNARK proof.

\printbibliography
\end{document}